\documentclass[12pt]{article}
\usepackage{amsmath}
\usepackage{graphicx}
\usepackage{enumerate}
\usepackage{times}
\usepackage{url} %
\usepackage[thm_section]{macros}

\usepackage[all=normal, paragraphs=tight, bibbreaks=tight, floats=tight,
mathdisplays=tight,bibnotes=tight]{savetrees}
\usepackage{amssymb}
\usepackage{bbm}
\usepackage{thm-restate}

\newcommand{\blind}{0}

\usepackage{titletoc}
\newcommand\DoToC{%
 \vskip1cm
 \startcontents[sections]
  \printcontents[sections]{l}{1}{\textbf{Contents}\vskip3pt\hrule\vskip5pt\setcounter{tocdepth}{2}}
  \vskip5pt\hrule\vskip5pt
}

\usepackage{geometry}
\newgeometry{margin=1in}

\newcommand\numberthis{\addtocounter{equation}{1}\tag{\theequation}}

\newcommand{\obs}{\text{obs}}

\begin{document}

\def\spacingset#1{\renewcommand{\baselinestretch}%
{#1}\small\normalsize} \spacingset{0}

\onehalfspacing

  \title{\bf \large Optimal Conditional  Inference in Adaptive Experiments\thanks{An extended abstract of this paper was presented at CODE@MIT
   2021. We thank Karun Adusumilli, Kei Hirano, Lucas Janson, Lihua Lei, Jonathan Roth,
   Brad Ross, David Ritzwoller, Pedro Sant'Anna, and participants at the Microsoft
   Research causal inference reading group, Stanford, CODE@MIT, the International Seminar
   in Selective Inference, and MLESC 2025 for comments.}}
  \author{Jiafeng Chen \hspace{.2cm}\\ Department of
   Economics, Stanford University\\ jiafeng@stanford.edu\footnote{Corresponding author} \\
   \quad \\ Isaiah
   Andrews \\
   Department of Economics,
   Massachusetts Institute of Technology \\ iandrews@mit.edu}
  \maketitle

\bigskip
\begin{abstract} 
We study batched bandit experiments and consider the problem of inference
 conditional on the realized stopping time, assignment probabilities, and target
 parameter, where all of these may be chosen adaptively using information up to the last
 batch of the experiment. Absent further restrictions on the experiment, we show that
 inference using only the results of the last batch is optimal.  When the adaptive
 aspects of the experiment are known to be location-invariant, in the sense that they are
 unchanged when we shift  all batch-arm means by a constant, we show that there is
 additional information in the data, captured by one additional linear function of the
 batch-arm means. In the more restrictive case where the stopping time, assignment
 probabilities, and target parameter are known to depend on the data only through a
 collection of polyhedral events, we derive computationally tractable and optimal
 conditional inference procedures.
\end{abstract}

\noindent%
\vfill

\newpage

\onehalfspacing

\section{Introduction} %
\label{sec:introduction}

Consider a \emph{batched} bandit experiment \citep
{zhang2020inference,hirano2023asymptotic}, where, at each batch $t = 1, 2,\ldots, T_0$,
an experimenter either continues or ends the experiment.  If they continue the
experiment, $n_t$ units are independently sampled and independently randomized among $K$
different arms. The
probability that a unit is assigned to arm $k$ in batch $t$ is denoted by $\Pi_
{tk}$, usually decided by bandit algorithms based on the history of outcomes. We let
$\Pi_{t} =\diag(\Pi_{t1},\ldots, \Pi_{tK})$ collect the assignment probabilities in batch
$t$. Observations assigned to arm $k$ have mean outcome $\mu_k$, collected in $\mu =
[\mu_1,\ldots,\mu_K]'$, and variance $\sigma_k^2$, collected in $\Sigma = \diag
(\sigma_1^2,\ldots,
\sigma_K^2)$.  Once the experiment ends, either at an experimenter-selected stopping time
($T<T_0$) or upon reaching the last batch ($T=T_0$), we are interested in statistical
inference on some linear combination of the arm means, $\eta'\mu$, where $\eta$
may depend on the data.

The iterative nature of the experiment provides multiple opportunities for adaptive,
data-driven decision-making.  First, treatment assignments may be selected using bandit
algorithms, which set $\Pi_ {t}$ using the results in batches $1,\ldots,t-1$.  For
instance, the experimenter might assign more units to arms which have produced good
outcomes in the past, as in Thompson sampling.  Second, the experimenter may choose
whether to continue the experiment---and thus the observed number of batches $T$---based
on how the results unfold.  For instance, the experimenter might stop the experiment early
if the results are either highly promising or overly discouraging. Third, the results of
the experiment may also inform what objects the experimenter targets for inference, and
hence the target parameter $\eta'\mu$.  For instance, the experimenter might be interested
in inference on average outcomes under an arm that performed especially well in the
experiment.

Each of these forms of adaptivity presents challenges for inference.
\cite{zhang2020inference} and \cite{hadad2021confidence} highlight that standard inference
 procedures that ignore the adaptive choice of $\Pi_ {t}$ can lead to invalid inference
 (e.g. under-coverage for confidence sets) in adaptive experiments, and propose
 alternative procedures which are valid in large samples provided $\eta$ and $T$ are
 fixed in advance.  Similarly, the large and growing literature on anytime-valid
 inference is motivated by the long-standing observation that adaptively chosen stopping
 times $T$ can lead to arbitrarily poor performance for conventional inference
 procedures, and proposes alternatives which guarantee performance in settings where
 $\eta$ is fixed (see
\citealt{ramdasetal2023anytime} for a recent review).  Finally the large literatures on
post-selection and selective inference \citep[e.g.][]{Berk2013,
FST15,andrews2019inference} highlight that the data-driven choice of a target parameter
can invalidate standard inference procedures, and proposes valid alternatives focused
primarily on static settings or, equivalently, settings where $\Pi_ {t}$ and $T$ are fixed
in advance.

In the batched setting we consider, there is a simple procedure available which ensures
valid inference for many different rules for choosing
$(T,\Pi_{1:T},\eta)=(T,\Pi_1,\ldots,\Pi_T,\eta),$ where we collectively
shorthand these rules as ``the experimental design.''  So long as the experimental design depends only
on data observed up to period $T-1$---requiring, for instance, that decisions to stop the
experiment be made one period in advance, inference based on only the last batch of the
experiment is guaranteed to be valid. 

This ``last-batch-only'' approach is in one sense quite restrictive, since it discards
information from all but the last batch of the experiment.  At the same time,
last-batch-only inference is quite flexible in that it requires no knowledge of how $
(T,\Pi_{1:T},\eta)$ are chosen, other than that $(T,\eta)$ must be chosen one batch in
advance.  Consequently, this approach can accommodate a variety of experimental designs,
including ones where decisions are made adaptively by human decision-makers whose
preferences and decision rules are not fully understood.

This paper examines the extent to which it is possible to improve last-batch-only
inference while continuing to allow a very flexible class of experimental designs. We cast
this problem as searching for optimal inference procedures that are valid
\emph{conditionally} on the experimental design $(T, \Pi_{1:T}, \eta)$.  Since the
last-batch-only procedure is a conditionally valid one, the optimal conditional inference
procedure must weakly dominate it. One can think of focusing on conditional inference as
looking for ``safe free lunches'' that inherit the flexibility of last-batch-only. Of
course, conditionally valid intervals are able to accommodate many adaptive choices at the price of
being somewhat longer than unconditionally valid ones \citep[e.g.,][]{zhang2020inference}.

We start with a simple negative result. If there are no restrictions on how the design $
(T,\Pi_{1:T},\eta)$ is chosen, it is impossible to improve upon last-batch-only
inference. This motivates searching  for natural restrictions on the
experimental design that allows for confidence sets that nontrivially use the first $T-1$
batches. 

Many algorithms for selecting $(T,\Pi_{1:T},\eta)$ depend only on the
\emph{contrasts} between different arms, and are unaffected if we increase the average
 outcome for all arms by a fixed amount.  This is the case, for instance, whenever the
 experimenter's decisions depend only on the difference in outcomes relative to a fixed
 treatment arm. We term experimental designs with this property ``location-invariant.''
 
Our core result is that, in the class of location-invariant experiments, a sufficient
statistic for $\mu$ is given by the mean outcomes in the last batch, together with a
scalar $L$ summarizing information from earlier batches. Intuitively, $L$
represents
information that is orthogonal to all pairwise arm differences over the first $T-1$
batches---and thus survives the conditioning. Conditional on these pairwise differences,
the distribution of the sufficient statistic is (asymptotically) multivariate Gaussian.
Thus, inference for $\eta'\mu$ based on these statistics is (asymptotically)  a simple generalized least squares problem, and is optimal among the class of procedures which are conditionally valid given  all pairwise differences between batch-arm means in the first $T-1$ batches.  In particular, our approach dominates inference based on the last batch alone.

Two aspects of this result are worth commenting on. First, we derive it under the
assumption that batch-arm means are exactly normally distributed (though results in
\cref{sub:conditional_inference_in_polyhedral_algorithms} extend to cases where batch-arm
means are linear exponential family). This is a restrictive when viewed as an exact
description of the data-generating process. Instead, we view it as an asymptotic
description for large parts of applied practice, because (i) any batched bandit experiment
is \emph{asymptotically} equivalent---in the Le Cam--Hajek sense with large batch
sizes---to our Gaussian model \citep {hirano2023asymptotic}, and (ii) many adaptive
experiments are batched in practice, due to implementation and cost
concerns.\footnote{See, for instance, \citet {zhang2020inference} and references therein,
``We focus on the batched setting because it closely reflects many of the problem settings
where bandit algorithms are applied. For example, in many mobile health   and online
education problems multiple users use apps / take courses simultaneously, so a batch
corresponds to the number of unique users the bandit algorithm acts on at once. The
batched setting is even common in online recommendations and advertising because it is
impractical to update the bandit after every action if many users visit the site
simultaneously.'' } Thus, so long as one considers a batched bandit experiment with large
batch sizes, our finite-sample optimality results for Gaussian data translate to---and
should be viewed as---\emph {asymptotic optimality} results for non-Gaussian data.
Moreover, the optimal procedure in the Gaussian limit experiment suggests a feasible
procedure applicable to data without distributional assumptions. Indeed, we show that
these feasible confidence intervals have uniformly correct asymptotic coverage over a
large class of non-Gaussian data generating processes.\footnote{These asymptotic results
were already present in Appendix C of the first draft of this paper (September, 2023,
arXiv:2309.12162v1). Independently, later work on bandit-experiment asymptotics
\citep{niu2025assumption} employs closely related proof techniques, showing asymptotics
for a wider class of statistics. See \cref{rmk:proof} for additional discussion. }

Second, because we allow all designs that satisfy location-invariance, our optimal
inference procedure conditions on all batch-arm mean differences. In some cases, however,
the information contained in $ (T,\Pi_{1:T},\eta)$ may be much coarser than all batch-arm
mean differences (though not in all cases: we show that if $(T, \Pi_{1:T}, \eta)$ is
generated by Thompson sampling, conditional validity given the design requires
conditioning on all batch-arm mean differences).  When the information contained in $
(T,\Pi_{1:T},\eta)$ is known to be coarser, we may exploit this knowledge to derive more
powerful conditional inference procedures. To illustrate, we consider settings where the
rules used to determine $(T,\Pi_ {1:T},\eta)$ are fully known and can be expressed in
terms of a finite number of linear-in-data inequalities (that is, settings where these $
(T,\Pi_{1:T},\eta)$ depend on the experimental results only through a collection of
polyhedral events).  This holds, for instance, when $\Pi_{1:T}$ is constructed using an
$\varepsilon$-greedy algorithm and $\eta$ selects the best-performing arm.  In such
settings, we characterize optimal median-unbiased estimators and equal-tailed confidence
intervals conditional on $(T,\Pi_ {1:T},\eta)$ in the limit experiment. These confidence
intervals are closely related to those in selective inference \citep{FST15}.

These results have implications for a wide range of empirical contexts. Adaptive
experiments are used in a variety of academic and commercial settings
\citep[e.g.][]{SBF17,Rafferty_Ying_Williams_2019,CGKQST22,bibaut2024demistifying}.
Moreover, recent theoretical results by \cite{HZZBGY20} show that batched experiments with
a surprisingly small number of batches can achieve the same rate guarantees as fully
adaptive experiments, in which the assignment probabilities may be different for each
observation, further motivating our focus on batched experiments. 

A potentially less obvious application of our results is to experiments with an initial
pilot stage.\footnote{We thank \blind{[anonymous] }bringing
this connection to our attention.} For the purposes of our analysis the pilot stage can be
treated as the first batch, and the main experiment as the second.  If the choice of
assignment probabilities and target parameters depend only on treatment-control contrasts
at the pilot stage, one can apply our results on inference for location-invariant
experimental designs and obtain more precise inference than is possible using the main
experiment alone.  Crucially, these results apply even when we do not have an exact
formula for how the design of the main experiment depends on the pilot, for instance
because design decisions were made on an ad-hoc basis after observing the results from the
pilot.

\Cref{sec:problem_setup} formally introduces the problem we consider in the context of the
normal model, and proves the optimality of last-batch-only inference absent further
restrictions on $(T,\Pi_{1:T},\eta)$.  \Cref{sub:leftover} introduces the class of
location-invariant experiments and derives the sufficient statistic for $\mu$ in this
class. 
\Cref{sub:asymptotics} shows that feasible analogs of the resulting inference procedures are uniformly asymptotically valid over a large class of data generating processes.
\Cref{sec:inference_in_the_gaussian_model} considers optimal conditional inference when $
 (T,\Pi_{1:T},\eta)$ depends on the data only through a collection of polyhedral events.
 Finally \cref{sec:simulation_evidence} provides simulation evidence on the performance
 of our procedures.  Proofs and additional results are provided in the appendix. 

\section{Problem setup} %
\label{sec:problem_setup}

Let $n_t$ denote the total number of observations in batch $t$, $n = \sum_{t=1}^T n_t$ the
total number of observations, and $c_t = n_t / n$ the fraction of observations assigned
in batch $t.$ We write $X_{tk}$ for the sample average outcome among those units assigned
to arm $k$ in batch $t$ and $X_t$ for the vector of these means, $X_t = [X_
{t1},\ldots, X_{tK}]'$.   We assume that the assignment probabilities $\Pi_{t}$ depend on
the data only through $X_{1:t-1}=(X_1,\ldots,X_{t-1}).$

Since $X_t$ is a vector of sample means,
when batch sizes are large, the central limit theorem implies that $X_t$ is approximately
Gaussian conditional on $\Pi_t,X_{1:t-1}$, with $X_{tk} \mid \Pi_t,X_{1:t-1} \sim \Norm
(\mu_k,\frac{\sigma_k^2}{n_t\Pi_{tk}}).$  For $V_t$  the diagonal matrix with $k$\th{}
diagonal element $\frac{\sigma_k^2}{c_t\Pi_{tk}}$ we thus have that, approximately,
\[ X_{t} \mid \Pi_t,X_{1:t-1} \sim \Norm\pr{\mu, \frac{1}{n} V_t},
\numberthis
\label{eq:gaussian}
\]
Since $\Sigma$  is consistently estimable, in \eqref{eq:gaussian} we treat
it (and thus $V_t$) as known. 

Exact Gaussianity \eqref{eq:gaussian} is a strong assumption. Our results on
general conditional inference \cref
{sub:conditional_inference_in_polyhedral_algorithms} extends to exponential family
models (\cref{rmk:expofam}), though within exponential families we view Gaussian as
likely the leading case. More importantly, we view the model
\eqref{eq:gaussian} as reflecting a focus on \emph{asymptotic} approximation for \emph
 {batched} bandits, rather than as a substantive parametric restriction. Indeed, the
 known-variance Gaussian model corresponds to the limit experiment for the batched bandit
 experiment under mild conditions 
 \citep{hirano2023asymptotic}. Since the statistical properties of any procedure in
  batched bandits can be matched asymptotically with a procedure under \eqref
  {eq:gaussian}, we derive optimal procedures under \eqref{eq:gaussian}. Thus, results
  under the Gaussian model characterizes {asymptotic} results for any batched bandit
  experiment, in the Le Cam--Hajek sense.

  To demonstrate that these procedures are indeed asymptotically valid, \Cref
  {sub:asymptotics}  verifies that the finite-sample analogue of our procedure converges
  to its counterpart under
\eqref{eq:gaussian}, uniformly over a large class of (not necessarily Gaussian)
 data-generating processes. For analytical convenience, outside of \cref
 {sub:asymptotics}, we maintain that $\Pi_ {tk} > 0$ almost surely for all $t \in
 [T], k\in [K]$.

\subsection{Inference problem}

Suppose that we observe data from batches $1,\ldots,T$ of the Gaussian experiment 
\eqref{eq:gaussian}, where the stopping decision is based on information available at
 $T-1$ (i.e. $\one(T\le t)$ is measurable with respect to $X_{1:t-1}$).  We
 are interested in inference on $\eta'\mu,$ where the target parameter is again
 measurable with respect to the information available at $T-1$, $\eta=\eta(X_
 {1:T-1},T).$

Since we require that stopping decisions and target parameters are determined by
$X_{1:T-1}$, inference based only on the last batch is valid: The
$z$-statistic confidence interval $\eta'X_T\pm z_{1-\frac{\alpha}{2}}
\sigma_{\eta}$ (where $z_{1-\frac{\alpha}{2}} = \Phi^{-1}(1-\frac{\alpha}{2})$ and
 $\sigma_{\eta}^2 = \frac{1}{n}\eta'V_T\eta$) has correct conditional coverage
\[
\P\left(\eta'\mu\in \bk {\eta'X_T\pm z_{1-\frac{\alpha}{2}}
\sigma_{\eta}} \mid T,\Pi_
{1:T},\eta\right)=1-\alpha.
\]
Moreover, $\eta'X_T$ is unbiased for $\eta'\mu$ even conditional on $T,$ $\Pi_{1:T}$, and
$\eta$: $\E[\eta'X_T \mid T,\Pi_{1:T},\eta]=\eta'\mu.$

Can we  construct a confidence set $C_{\eta}$ that improves on last-batch-only inference
while maintaining the same flexibility? Specifically, if we would like confidence
sets $C_\eta(\alpha)$ that maintain validity \emph{conditional} on the experimental design
(stopping time,
assignment probability, and target parameter) \[
    \P\pr{
        \eta'\mu \in  C_{\eta}(\alpha) \mid T,\Pi_{1:T},\eta
    } \ge 1-\alpha \text{ almost surely}, \quad \alpha \in (0,1).
    \numberthis
    \label{eq:conditional_requirement}
\] Can we do better than last-batch-only? The answer turns out to be \emph{no}, if $\eta$
 can be an arbitrary function $\eta(\cdot)$ of the history $X_{1:T-1}$.
\begin{restatable}{lemma}{lemmaimprovability}

    \label{lemma:improvability}

    Suppose \eqref{eq:conditional_requirement} holds for all measurable $\eta(\cdot)$,
    then for all fixed $\eta \in \R^K$\[
        \P\pr{
            \eta'\mu \in C_{\eta}(\alpha) \mid X_{1:T-1}
        } \ge 1-\alpha  \mbox{ almost surely}.
    \]
\end{restatable}

\Cref{lemma:improvability} shows that if $\eta(\cdot)$ is entirely unrestricted, then
any confidence set which is conditionally valid in the sense of 
\eqref{eq:conditional_requirement} must also be valid conditional on $X_{1:T-1}.$  In
 fact, requiring conditional coverage given $\eta$ is already sufficient to obtain this
 conclusion. This effectively means we cannot use information from the first $T-1$
 batches of data, so it is impossible to improve on the power of last-batch-only
 inference without losing some degree of robustness. In other words, for the first $T-1$
 batches to provide useful information, we must be able to restrict the design. 

To be sure, \eqref{eq:conditional_requirement} is stronger than some coverage
requirements targeted in the literature, which do not condition on  $ (T, \Pi_
{1:T}, \eta)$ \citep{zhang2020inference,hadad2021confidence}. Correspondingly,
imposing \eqref {eq:conditional_requirement} will tend to yield longer confidence sets.
On the other  hand, relative to the anytime-valid literature, we require weaker coverage
(specifically, a weaker form of time-uniformity), and consequently tend to deliver tighter confidence
intervals.\footnote{For example, the width of $(1-\alpha)$-asymptotic confidence
sequences in
\citet{waudby2024time} scales as $\sqrt{\log (\sqrt{n}/\alpha) \frac{1}{n} }$
 (Theorem 2.2), whereas our confidence sets scale as the usual $1/\sqrt{n}$.}

Whether imposing \eqref {eq:conditional_requirement} is appealing depends on whether one
views conditional coverage as important.  In settings where $\eta$ is adaptively
chosen, \citet{fithian2014optimal} advocate conditioning on $\eta$ on the grounds that
that ``The answer must be valid, given that the question was asked.''  More broadly for
confidence sets that are unconditionally, but not conditionally, valid, there necessarily
exist values of $(T, \Pi_{1:T}, \eta)$ at which the confidence set under-covers some true
parameter.\footnote{See
 \cref{asub:zhang_bet_proof} for an example with \citet{zhang2020inference}.} Thus, if we
  wish to avoid the situation where a reader could reasonably disbelieve our coverage
  claim after learning $(T, \Pi_{1:T}, \eta),$ we should focus on conditional coverage.
  On the other hand, if we are unconcerned with this possibility we are free to consider
  the (strictly larger) class of unconditionally valid procedures.

 This tradeoff between a more robust form of Type I error control and power is an
 important but delicate question present in many contexts in statistics. Without
 taking a stand on on the ``correct'' form of error control in general, 
 our motivation for focusing on \eqref
 {eq:conditional_requirement} stems from a desire to search for ``safe free lunches'' over
 the simple and intuitive procedure last-batch-only. Since last-batch-only satisfies
   \eqref{eq:conditional_requirement}, optimal procedures satisfying 
   \eqref{eq:conditional_requirement} by definition improves on last-batch-only. 

To make progress in light of \cref{lemma:improvability}, we consider two sets of
restrictions on the design. First, we observe that many bandit algorithms are  \emph
{location-invariant}. That is, they have the property that $\Pi_t(X_{1:t-1})$ is
invariant to adding a constant $h \in
\R$ to every batch-arm mean $X_{sk}$. If the adaptive choices of inferential target and
 stopping time are similarly location-invariant, as is often the case, then we can
 condition on less information, and construct procedures that dominate using solely the
 last batch. Such a procedure does not require knowledge of the precise allocation
 algorithm nor of $\eta(\cdot)$ and $T$, beyond location-invariance.

Second, if the experimental design depends solely on a lower-dimensional but known set of
statistics, then we can also design optimal conditional inference procedures conditioning
on these statistics. We show that such procedures are particularly tractable for a large
class of discrete assignment algorithms that we call polyhedral algorithms. Since these
procedures require less stringent conditioning, they are more powerful than conditional
inference procedures that only use location-invariance.

\section{Conditional inference for location-invariant algorithms}
\label{sub:leftover}

This section focuses on location-invariant designs. We show that there is a simple
conditional procedure that assumes only location-invariance and improves upon
last-batch-only inference. To state these results, we first formally define what we mean
by location-invariance.

\begin{defn}
A function $f(X_1,\ldots, X_t)$ is \emph{location-invariant} if \[f(X_1 + h1_K, \ldots,
X_t +
h1_K) = f (X_1,\ldots, X_t)\] for all $(X_1,\ldots, X_t) \in (\R^K)^t$ and $h\in \R$, where $1_K\in\mathbb{R}^K$ is the vector of ones.
\end{defn}

\begin{defn}
An assignment algorithm is \emph{location-invariant} if each batchwise probability $\Pi_t
= \Pi_t(X_1,\ldots, X_ {t-1})$ is location-invariant. 
\end{defn}
\begin{defn}
A stopping time $T$ is \emph{location-invariant} if whether batch $t$ is the last batch is determined by $X_{1:t-1}$ via a location-invariant function. That is, $\one(T > t) = \Xi_t(X_1,\ldots, X_{t-1})$ and $\Xi_t(\cdot)$ is location-invariant for all $t = 1,\ldots, T_0$.
\end{defn}
\begin{defn}
$\eta$ is location-invariant if $\eta = \eta(X_1,\ldots, X_{T-1}; T)$ and, for all $t
=1,\ldots, T_0$, each $\eta(X_1,\ldots, X_{t-1}; t)$ is location-invariant in its first
$t-1$
arguments.
\end{defn}

For concreteness, let us first introduce two leading location-invariant assignment
algorithms. Let \[ W_t = \pr{\sum_{s=1}^t V_s^{-1}}^{-1} \sum_{s=1}^t V_s^{-1}
X_s
\]
be the inverse-variance weighted batch-arm means, and let \[
\Omega_t = \frac{1}{n}\pr{\sum_{s=1}^t V_s^{-1}}^{-1} = \frac{\Sigma}{n} \pr{\sum_{s=1}^t
c_s \Pi_s}^{-1} \numberthis \label{eq:w_omega}
\]
be the posterior variance of $\mu \mid W_t$ under a flat prior for $\mu$. 
\begin{exsq}[Thompson sampling]
\label{ex:thompson}
Let $Q(\nu, \Lambda)$ denote the vector of Gaussian orthant probabilities \[
Q_k(\nu, \Lambda) = \P_{X \sim \Norm(\nu, \Lambda)}\pr{X_k \ge \max_{\ell} X_\ell}.
\numberthis \label{eq:orthant}
\]
Consider a Bayesian experimenter with a flat prior for $\mu$. 
The Bayesian's posterior distribution for $\mu$ after
observing $X_{1:t}$ is $
\mu \mid X_{1:t} \sim \Norm\pr{
    W_t, \Omega_t
}.
$
Thompson sampling then sets $\Pi_{t+1} = Q\pr{ W_t, \Omega_t }.$  Since $Q(\nu+h, \Lambda) = Q(\nu,
\Lambda)$, Thompson sampling is location-invariant.
\end{exsq}

\begin{exsq}[$\varepsilon$-greedy]
\label{ex:egreedy}
For some $0 < \varepsilon < 1/K$, the algorithm sets $\Pi_{tk} = 1-(K-1)\varepsilon$ if
$W_{tk}$ is the largest entry in $W_{t}$. Otherwise, the algorithm sets $\Pi_{tk} =
\varepsilon$. This algorithm is location-invariant since it only involves ordinal
comparisons of $W_{tk}$.
\end{exsq}

Location-invariance implies that  for each $t\in\{1,...,T_0\}$, the design $(\one(T\le
 t),\Pi_{1:t},\eta)$ is measurable with respect to  batch-arm differences against
 $X_{1K}$:
\[
\Delta X_{1:t-1} \equiv \br{X_{sk} - X_{1K} : s=1,\ldots, t-1; k = 1,\ldots, K-1}.
\]
Consider the following statistic, which represents some information from the first $T-1$
batches ``left over'' from conditioning on $(T,\Delta X_{1:T-1})$: \[ L =
\sum_{t=1}^{T-1} \sum_{k=1}^K
c_t \frac{ \Pi_{tk}}{\sigma_k^2} X_{tk}. 
\] Since $c_t n \Pi_{tk}$ is the number of units assigned to arm $k$ in batch $t$, the
 coefficient $c_t \frac{ \Pi_{tk}}{\sigma_k^2}$ is then the proportional to the inverse
 variance of $X_{tk}$. Thus we may view $L$ as the precision-weighted \emph{sum} of all
 batch-arm means in the first $T-1$ batches. This particular weighted sum turns out to be
 orthogonal to random variables measurable with respect to $\Delta X_{1:t-1}$.\footnote
 {It is illustrative to note that for $Z \sim \Norm(0, \diag(\sigma_1^2,
\sigma_2^2)),$ the precision-weighted sum $Z_1/\sigma_1^2 + Z_2/\sigma_2^2$ is independent
 of the difference $Z_1 - Z_2$.}

It turns out that $L$ is the \emph{only} information left over from conditioning, in the
sense that $(L, X_T)$ is sufficient for $\mu$ with respect to the conditional distribution
of the data.%

\begin{restatable}{theorem}{thmleftover}
\label{thm:leftover}
    Under the preceding setup, assuming $\Pi_{t} > 0$ for all $t$ almost surely, \[
\colvecb{2}{L}{X_T} \mid (T, \Delta X_{1:T-1}) 
\sim \colvecb{2}{L}{X_T} \mid (\Pi_{2:T}, T) 
\sim
\Norm\pr{
    \colvecb2{\lambda'\mu}{\mu}, \begin{bmatrix}
        \frac{\lambda'1}{n} & 0 \\
        0 & \frac{1}{n} V_T
    \end{bmatrix}
},\numberthis \label{eq:dist_suff_stat_diff}
    \]
where $\lambda = (\lambda_1,\ldots, \lambda_K)'$ and $\lambda_k = \frac{1}{\sigma_k^2}\sum_
{t=1}^{T-1}
     \frac{n_t}{n}\Pi_{tk}$.
Moreover, $(L, X_T)$ is sufficient for $\mu$ with respect to the conditional
        distribution $X_{1:T} \mid (T, \Delta X_{1:T-1})$.
        
\end{restatable}

Therefore, we can base inference for $\mu$ on \eqref{eq:dist_suff_stat_diff}. Optimal
inference in \eqref{eq:dist_suff_stat_diff} is also optimal conditional inference given the
finer information set $\Delta X_{1:T-1}$. 
 The minimal sufficient statistic for $\mu$ under
\eqref{eq:dist_suff_stat_diff} is the weighted least-squares coefficient
\begin{align*}
S^\star %
&= X_T + \frac{ V_T \lambda}{\lambda'1 +  \lambda' V_T \lambda}(L - \lambda'X_T)
\numberthis
\label{eq:wls_coef}
\\
\text{ where } S^\star \mid (\Delta X_{1:T-1}, T) &\sim \Norm\pr{
    \mu, \frac{1}{n} \pr{V_T - \frac{V_T  \lambda \lambda' V_T}{\lambda'1 +
    \lambda'V_T \lambda}}
} \numberthis \label{eq:asymp_variance}.
\end{align*} Optimal conditional inference for $\eta'\mu$ is thus simply based on the
 Gaussian statistic $\eta' S^\star$. For instance, the optimal uniformly most accurate unbiased conditional confidence interval for $\tau$ is the z-statistic interval centered at $\eta' S^\star$. \cref{sub:asymptotics} shows that   the finite
 sample analogue of $\eta'S^\star$ has asymptotic behavior analogous to \eqref
 {eq:asymp_variance} and produces confidence intervals that cover the random parameter
 $\tau = \eta'\mu$ conditional on $(T,\Pi_{1:T},\eta)$.

\subsection{Gain relative to last-batch inference}

\Cref{thm:leftover} shows that under location-invariance there is one usable piece of
information beyond the last batch, namely the statistic $L$.  It is therefore
interesting to understand to what extent this additional information helps us with
inference on $\mu$.

\newcommand{\rc}{r^{(c)}}
\newcommand{\rt}{r^{(t)}}

Let $\rc = V_T \lambda$. Note that $\rc$ collects the ratio of cumulative sample 
size and last period sample size for each arm: $\rc_k = \frac{\sum_{t=1}^{T-1} n_t\Pi_
{tk}}{n_T \Pi_{Tk}}.$ Similarly, let $\rt = \rc + 1$ collect the ratio of total sample
size to last period sample size for each arm. Let $q = \sum_{t=1}^{T-1} n_t \Pi_{t}1_K$
collect the cumulative sample size. Then we have that \[
\var(S_k^\star \mid \Delta X_{1:T-1}, T) = \underbrace{\frac{V_T}{n}}_{\var(X_T \mid
\Pi_T)} \pr{I - \frac{q \Sigma^{-1}
(\rc)'}{q' \Sigma^{-1} \rt}}.
\]
Thus, the reduction in variance relative to using solely the last batch depends on the
alignment of $\eta$ to the matrix $\frac{q \Sigma^{-1}
(\rc)'}{q' \Sigma^{-1} \rt}$, which collects information about variance-weighted relative
sample sizes. The relative improvement is greatest when $\eta$ is proportional to
$\lambda$, since then $L$ provides the greatest amount of information.

\subsection{Conditioning on $\Delta X_{1:T-1}$ in Thompson sampling }

We have so far studied inference conditional on all of the differences $\Delta X_
{1:T-1}$, since $(L, X_T)$ is sufficient with respect to $X_{1:T} \mid (T, \Delta X_
{1:T-1})$. While this allows for any location-invariant experimental design, one might
wonder if this conditioning is excessive.  If we  had more restrictions on how $(T,\Pi_
{1:T},\eta)$ were generated and  our only goal were to ensure conditional coverage \eqref
{eq:conditional_requirement}, perhaps it would suffice to condition on a coarser set of
statistics, which could lead to higher power.  \cref
{sec:inference_in_the_gaussian_model} explores this possibility.

Here, we show instead that conditioning on $\Delta X_{1:T-1}$ is sometimes necessary for
conditional coverage
\eqref{eq:conditional_requirement}, even when the rules generating $(T,\Pi_{1:T},\eta)$
are perfectly known.  Specifically, we show in the case of Thompson sampling
(\cref{ex:thompson}), for fixed $T$ and $\eta$,  efficient conditional inference given
$\Delta X_{1:T-1}$ is also efficient given $\Pi_{1:T}$. This is because $\Pi_{1:T}$ in
Thompson sampling contains rich information about $\Delta X_{1:T-1}$.

\begin{restatable}{prop}{propovercondition}
\label{prop:overcondition} 
Under the setup of \cref{thm:leftover}, suppose, for all $t$, $\Xi_{t}(X_1,\ldots,
    X_{t-1})$ is measurable with respect to $\Pi_{2:t}$. Under Thompson sampling, $(L,
    X_T)$ is sufficient for $\mu$ with respect to the conditional distribution
    $X_{1:T} \mid \Pi_{1:T}, T$.
\end{restatable}

\Cref{prop:overcondition} shows that $(L, X_T)$ is sufficient for $\mu$ with respect to
the distribution $X_{1:T} \mid (\Pi_{2:T}, T)$. Thus, optimal inference based on $(L,
X_T)$ is optimal for $\eta'\mu$. The proof relies on inversion results from the
econometrics of discrete choice \citep{hotz1993conditional,norets2013surjectivity} in
order to show that $\Pi_{1:T}$ reveals sufficiently rich information on $\Delta
X_{1:t-1}$.

\section{Uniform asymptotics for location-invariant algorithms}
\label{sub:asymptotics}

Our focus on the Gaussian model is motivated by the fact that it serves as an
approximation to asymptotic results for any batched bandit setting, without
distributional assumptions. To more formally connect the two, this section shows that the
finite-sample analogue of the inference procedure following \cref{thm:leftover} is
uniformly asymptotically valid over a large class of data distributions and experimental
designs. 

To state these
 results, we first introduce the non-Gaussian version of our setting.
 Fix a sequence $(n, n_1,\ldots, n_{T_0})$ indexed by $n$ such that $n_t / n \to c_t \in (0,1)$ and
 $\sum_t n_t = n$. Let $\mathcal I_t \subset [n]$ collect the indicies of the $n_t$
 individuals in batch $t$. Prior to each batch, assignment probabilities $\Pi_
 {tn} = \diag(\Pi_{t1n},\ldots, \Pi_{tKn})$ are determined by the assignment algorithm,
 as a function of past data. In batch $t$, the $n_t$ participants are assigned to one of
 $K$ arms. The indicator for individual $i$'s assignment is $D_i = [D_{i1},\ldots, D_
 {iK}]'$, which follows a categorical distribution with probabilities $\Pi_{tn}$ and are
 independent across $i$. Let $N_{tn} = [N_{t1n},\ldots, N_{tKn}]' = \sum_
 {i \in \mathcal I_t} D_i$. Let $\hat\Pi_{tkn} = N_{tkn} / n_t $ be the realized
 frequency of samples to arm $k$ in batch $t$. Let $\hat\Pi_{tn} =
\diag(\hat\Pi_{t1n},\ldots, \hat\Pi_{tKn}).$

 Let $X_i^{\obs} = \sum_k D_{ik} X_i(k)$ be the
outcome observed for $i$---equal to the potential outcome at the assigned treatment---and  
let \[
X_{tkn} = \frac{1}{N_{tkn} \vee 1} \sum_{i \in \mathcal I_t} D_{ik} X_i^{\obs}
\]
be the batch-arm empirical mean. The definition for
$X_ {tkn}$ builds in the convention that an empty
mean is zero. Let \[
\hat \sigma_{tkn}^2 =
\frac{\sum_{s=1}^t \sum_ {i\in \mathcal I_s} D_{ik} (X_i^\obs
    -W_{t, k, n})^2}{\sum_{s=1}^t N_{skn}} \numberthis \label{eq:def_sigma2}
\]
be an estimate of the arm variance using data up until batch $t$. Let $\hat \sigma_{kn}^2 =
\hat \sigma_{Tkn}^2$, $\hat\Sigma_{tn} = \diag(\hat\sigma_{t1n}^2, \ldots, \hat\sigma_
{tKn}^2)$ and $\hat\Sigma_n = \diag(\hat \sigma_{1n}^2,\ldots, \hat\sigma_{Kn}^2)$.

Let the realized stopping time be $T_n$ and suppose the experimenter wishes to conduct inference for a possibly
data-dependent parameter $\tau_n = \eta_n'\mu$. We let \[
\lambda_n = \sqrt{n} \sum_{t=1}^{T_n-1} \frac{n_t \hat \Pi_{tkn}}{n \hat\sigma^2_{kn}}
\quad
L_n = \sqrt{n} \sum_{t=1}^{T_n-1} \sum_{k=1}^K \frac{n_t \hat \Pi_{tkn}}{n \hat\sigma^2_
{kn}} X_{tkn}
\]
be scaled analogues of $\lambda$ and $L$. Finally, let \[
S_n^\star = \pr{\frac{\lambda_n \lambda'_n }{\sqrt{n} \lambda'_n1} + \frac{n_{T_n}}{n}
\hat\Sigma^{-1}_n \hat \Pi_{T_n,n} }^+
\pr{\frac{\lambda_n}{\lambda_n'1} \frac{L_n}{\sqrt{n}} + \frac{n_{T_n}}{n} \hat\Sigma^
{-1}_n
\hat\Pi_
{T_nn} X_{T_nn}}\numberthis \label{eq:real_wls_coef}
\]
be an analogue of \eqref{eq:wls_coef}.\footnote{Here,
to conform with our asymptotic results and to accommodate $\Pi_{tk} \approx 0$, we
study the following parametrization of
\eqref{eq:dist_suff_stat_diff}
\[
\colvecb{2}{L}{\sqrt{n_t} \Pi_T X_T} \mid \Pi_{2:T} \sim \Norm\pr{
    \colvecb{2}{\lambda'\mu}{\sqrt{n_t}\Pi_t \mu}, \begin{bmatrix}
        \frac{\lambda'1}{n} & 0 \\ 0 & \Pi_T \Sigma
    \end{bmatrix}
}
\]
which we again treat as a weighted least-squares problem for $\mu$.
\eqref{eq:wls_coef} follows
from applying the Sherman--Morrison identity in the case that $\Pi_T > 0$. } The estimator
for $\tau_n$ is $\hat\tau_n = \eta_n' S_n^\star$, with estimated asymptotic standard error
\[\hat\sigma^2_{\tau,n} = \eta_n'
\pr{\frac{\lambda_n \lambda'_n /
\sqrt{n}}{\lambda'_n1} + \frac{n_{T_n}}{n}
\hat\Sigma^{-1}_n \hat \Pi_{T_n, n} }^+ \eta_n. \]

Our main asymptotic result is that, uniformly over a class $\mathcal P$ in the sense of \citet{andrews2011generic}, $
\sqrt{n} \hat\sigma^{-1}_{\tau,n} (\eta_n'S_n^\star - \tau_n)  \dto Z
$
where $Z$ is conditionally standard Gaussian given (the limiting analogues of) $(T_n, \Pi_{1:T_n}, \eta_n)$.
Thus, $(1-\alpha)$ two-sided confidence intervals $
\mathrm{CS}_{n}(\alpha) \equiv \eta_n'S_n^\star \pm \Phi^{-1}(1-\alpha/2) \cdot
\hat\sigma^2_{\tau,n}
$
have exact asymptotic conditional size for the parameters $\tau_n$:\footnote{For
expositional clarity, the main text \eqref{eq:uniform_inference} presents the conditional
coverage as conditioning on $T$ and $\eta$. \Cref{thm:appendix_main} in the appendix
additionally shows conditional coverage given $\Pi_{1:T}$.} \[
\limsup_{n\to\infty}\sup_{P \in \mathcal P} \, \abs[\big]{P\pr{
    \tau_n \in \mathrm{CS}_{n}(\alpha)
 \mid T_n = T, \eta_n = \eta } - \alpha} P(T_n = T, \eta_n = \eta) = 0.
\numberthis
\label{eq:uniform_inference}
\]

Such a statement is predicated on assumptions about the class $\mathcal P$, the assignment
algorithm for generating $\Pi_{1:T_0,n}$, $T_n$, and $\eta_n(\cdot)$. We assume that
members of
$\mathcal P$ have second moments bounded away from $0$ and $\infty$ and uniformly bounded
fourth moments.
\begin{as}
\label{as:P_assumptions}
$\mathcal P$ is a class of distributions satisfying the following conditions indexed by
constants $C_1, C_2 > 0$: For all $P \in \mathcal P$, let $\sigma_k^2(P) = \var(X_i(k))$ be the variance of the $k$\th{} potential outcome \begin{enumerate}
    \item For all $k$, $\frac{1}{C_1} \le \sigma_k^2(P) \le C_1$.\footnote{The upper bound
     is redundant given item 2 in the assumption, but we will keep it for convenience.}
    \item $\E_P[\norm{X}^4] < C_2$ where $X \sim P$ is a vector of potential outcomes.
\end{enumerate}
\end{as}

This section restricts to $T_n, \eta_n, \Pi_{t+1, n}$ that depend on $X_{1:T_0}$ only
through the cumulative mean vector\footnote{Our proof in \cref
{sub:asymptotics-proof} considers weaker conditions.} \[ W_{tkn} = \sum_{s=1}^{t}
\frac{n_s \hat \Pi_{skn}}{\sum_{r=1}^{t} n_r \hat\Pi_{rkn}} X_{skn}.
\]
Note that this covers the Thompson sampling and $\varepsilon$-greedy examples discussed above.
Let \begin{align*} 
W_{tn} = \pr{\sum_{s=1}^{t} \Pi_{sn}}^{-1} \sum_{s=1}^t
\Pi_ {sn} X_{sn} \quad 
\Omega_{tn} = \hat \Sigma_{n} \pr{\sum_{s=1}^t \frac{n_t}{n} \hat \Pi_{sn}}^{-1} 
\end{align*}
be the analogues of $W_t, \Omega_t$ in \eqref{eq:w_omega}. We assume that $
\one(T_n > t+1) = \Xi_{t+1}\pr{
\sqrt{n} W_{tn}, \Omega_{tn} } $ for some known function $\Xi_{t}.$ Likewise, we assume
 that the assignment probability is generated through some function $\kappa_t$: \[
\Pi_{t+1,n} = \kappa_{t+1,n}\pr{X_{1:t,n}, \hat\Pi_{1:t,n}, \hat\Sigma_{tn}, n_{1:t}} =
\kappa_ {t+1} \pr{
    \sqrt{n} W_{tn}, \Omega_{tn}
}.
\numberthis
\label{eq:simplied_mechanism}
\]
We assume that both $\Xi_{tn}$ and $\kappa_{t+1, n}$ are location-invariant in $W$.
Moreover, we 
require these functions to be suitably continuous, so that they converge weakly when their
arguments converge weakly. All of the following assumptions are relative to some $\epsilon
> 0$. \Cref{sub:kappa-checks} describes some concrete algorithms that satisfy the
following assumptions.

\begin{defn}
\label{def:continuity}

Consider a function $g(w, \Omega)$, we say that $g$ is
\emph{location-invariant and adequately continuous} (LIAC) if the following is true:
\begin{enumerate}[wide]
    \item For all $c \in \R$, $g(w + c1_K , \Omega) = g(w, \Omega)$
    \item For every nonempty $J \subset [K]$,
every diagonal matrix $\Omega$ with entries within $[\epsilon, 1/\epsilon]$, Lebesgue-\emph{almost} every
$(w_j: j \in J) \in \R^{|J|}$,  and every sequence $\R^K \ni w_m \to w_\infty$ and
$\Omega_m \to \Omega$ where \[w_{\infty, j} =
\begin{cases}
    w_j, &\text{ if } j\in J\\
    -\infty & \text{ otherwise}.
\end{cases}\] and $\Omega_m$ are positive-definite diagonal matrices, we have that the
limit \[g(w_{\infty}, \Omega) = \lim_{m\to\infty} g(w_m, \Omega_m) \numberthis
\label{eq:adequately_continuous_statement}\] exists and depends only on $w_\infty, \Omega$.
\end{enumerate}
\end{defn} 

We remark on the continuity requirement \cref{def:continuity}(2). Quantities in the
adaptive experiment depend on $\sqrt{n} W_{tn}$ and $\Omega_{tn}$. We expect that $\sqrt
{n}W_{tn} - \sqrt{n}\mu$ converges in distribution to some continuous distribution, and
$\Omega_{tn}$ converges in probability to some matrix $\Omega$. We would like to impose
enough regularity so that $g(\sqrt{n} W_{tn}, \Omega_{tn})$ converges to a suitable
limit. By location invariance, letting $\bar \mu = \max_k
\mu_k$, we can replace $\sqrt{n}W_{tn}$ with $\sqrt{n}W_{tn} - \sqrt{n} \bar\mu  = (
\sqrt{n}
\mu - \sqrt{n}\bar\mu) + \sqrt{n}(W_{tn} - \mu)$. Given a suitable sequence of
data-generating processes $P_n$ with $\mu = \mu(P_n)$, (sub)sequences of $\sqrt{n}W_{tn} -
\sqrt{n}
\bar\mu$ have entries that either tend to a continuously distributed random variable
 (for coordinates where $\sqrt{n} (\mu_k-\bar\mu) = O(1)$) or diverge to $-\infty$. Thus,
 the continuity of $g$ needs to accommodate entries diverging to $-\infty$ as well. On
 the other hand, the requirement that the limit exists for \emph{almost} every $w_j$
 allows $g$ to be discontinuous on measure-zero sets, and accommodates examples such as $
  \varepsilon$-greedy. 

\begin{as}
\label{as:Xi_continuity}
For all $t \in \br{2,\ldots ,T_0-1}$, $\Xi_t(\cdot)$ is LIAC, and $\Xi_{t} \le \Xi_{t-1}$
almost surely. Moreover, $\Xi_1 = 1$ and $\Xi_{T_0} = 0$.
\end{as}

\begin{as}
\label{as:kappa_continuity}

For all $t \in \br{2,\ldots, T_0-1}$, \[\kappa_{t+1} = \underbrace{\Xi_{t}\Xi_{t+1}}_
{\one(T_n > t+1)} \kappa_{t+1,
(0)} + \underbrace{(\Xi_
{t} -
\Xi_{t+1})}_{\one(T_n = t+1)} \kappa_{t+1, (1)} + \underbrace{(1-\Xi_{t})}_{\one(T_n \le
t)}
\Pi_1.
\numberthis
\label{eq:kappa_decomposition}
\] and $\kappa_{t+1, (1)}, \kappa_{t+1, (0)}$ are LIAC functions of $\sqrt{n} W_{tn},
\Omega_{tn}$. Moreover, $\kappa_{t+1,
(0),k} (w, \Omega) > 0$  when $w_k > 0$, and $\Pi_1$ is fixed and its entries are bounded
below by $\epsilon$.
\end{as}

The expression \eqref{eq:kappa_decomposition} means that the assignment probability
algorithm $\kappa_ {t+1}$ takes the form of a contingency plan. If the experiment
continues, then the next-batch probabilities are prescribed by $\kappa_{t+1,
(0)}$. Otherwise, they are prescribed by $\kappa_{t+1, (1)}$. By convention, we say that
$\kappa_{t+1} = \Pi_1$ when the experiment has already stopped.

For the pseudoinverse $\pr{\frac{\lambda_n \lambda'_n }{\sqrt{n} \lambda'_n1} + \frac{n_T}{n}
\hat\Sigma^{-1}_n \hat \Pi_{Tn} }^+$ to converge weakly when its arguments do, we
additionally require that the rank of $\frac{\lambda_n \lambda'_n }{\sqrt{n} \lambda'_n1} + \frac{n_T}{n}
\hat\Sigma^{-1}_n \hat \Pi_{Tn}$ is well-behaved. A sufficient condition is to assume that
 the range of the assignment probabilities of the last batch excludes $
 (0, \epsilon)$. This can be implemented by \emph{pruning} away small probabilities and
 set them to zero, and redistributing the excess mass. Doing so ensures that the
 arm-level sample sizes in batch $T$ either diverges to infinity or is eventually zero.

\begin{as}
\label{as:kappa_pruning}
$\kappa_{t+1, (1)}$'s range excludes $(0, \epsilon)$ on every coordinate.
\end{as}

Lastly, we consider a simple form of $\eta_n$. We assume that $\eta_n$ takes finitely many
values, but they can be chosen with respect to some rich information.

\begin{as}
\label{as:eta_pruning}
$\eta_n$ takes finitely many values depending on $T_n$, the locations of the zero entries
in $\Pi_{T_n, n}$, and the ordering of $W_{T_n, n}$: Let $S_K$ be the set of permutations
of $K$ symbols,
\[\eta_n = \sum_{\varsigma \in S_K: W_{T_n, n, \varsigma(1)} < \cdots < W_{T_n, n, \varsigma(K)}
} \sum_{t=1}^{T_0} \sum_{E \subset [K]} 
    \one\pr{T_n = t, \br{k : \Pi_{T_n, n, k} = 0} = E} \eta_{\varsigma, t, E}.
\]
Moreover, $\norm{\eta_{\varsigma, t, E, k}} > \epsilon$ and $\eta_{\varsigma, t, E, k} =
0$ whenever $k \in E$.

\end{as}

\Cref{as:eta_pruning,as:kappa_continuity,as:Xi_continuity,as:kappa_pruning} are stronger
 than necessary. We state weaker conditions as
 \cref{as:eta_assumptions,as:kappa_cmt,as:T_assumptions,as:weighting_assumptions,as:pruning_general},
 and verify that the weaker conditions are implications of the stronger ones. Under these
 assumptions, we formally justify \eqref{eq:uniform_inference}. The following is a
 corollary of a more general result \cref{thm:appendix_main} in the appendix.

\begin{restatable}{theorem}{thmmaintext}
\label{thm:maintext}
Under \cref{as:P_assumptions,as:eta_pruning,as:kappa_continuity,as:Xi_continuity,as:kappa_pruning},
level-$(1-\alpha)$ two-sided confidence intervals $
\mathrm{CS}_{n}(\alpha) \equiv \eta_n'S_n^\star \pm \Phi^{-1}(1-\alpha/2) \cdot
\frac{\hat\sigma_{\tau,n}}{\sqrt{n}}
$
have exact conditional asymptotic size for the parameters $\tau_n$: For all $\eta, T$,  \[
\liminf_{n\to\infty}\inf_{P \in \mathcal P} \, \abs[\bigg]{P\pr{
    \tau_n \in \mathrm{CS}_{n}(\alpha)
 \mid \eta_n = \eta, T_n = T} - \alpha} P(\eta_n = \eta, T_n=T) = 0.
\]
\end{restatable}

\begin{rmksq}[Proof strategy and relation to \citet{niu2025assumption}]
\label{rmk:proof}
  The proof of \cref{thm:maintext} largely concurs with recent independent work by
\citet{niu2025assumption}, who show an analogue of \cref{thm:main_convergence} (Theorem
C.1 in an earlier draft\footnote{\url{https://arxiv.org/pdf/2309.12162v1}} of this paper)
in the two-batch, two-arm case, though for a broader class of statistics than we consider.

Both rely on the following idea. Let $Y_{tn} = \sqrt{n_t} \hat \Pi_{tn} (X_{tn} - \mu)$ be
a properly normalized statistic and let $Y_{t} = \sqrt{n_t}\Pi_{t}(X_t-\mu)$ be its
Gaussian counterpart. The key is to build up a joint convergence  \[ (Y_{1n}, \hat
\Pi_{1n}, Y_{2n}) \dto (Y_1, \Pi_1, Y_2), \quad Y_2 \mid Y_1 \sim \Norm(0, \Pi_2\Sigma)
\numberthis,
\label{eq:joint_convergence}
 \]
 from the convergence $ (Y_{1n}, \hat \Pi_{1n}) \dto (Y_1, \Pi_1) $, which is easily
 implied by the central limit theorem and continuous mapping theorem.  By repeatedly
 adding on statistics from later stages to this convergence, like
 \eqref{eq:joint_convergence}, we can inductively show that all relevant statistics
 converge to their Gaussian-model counterpart over $T$ batches.

 We prove \eqref{eq:joint_convergence} by defining an intermediate quantity $Y_2^* \mid
 \hat\Pi_{1n} \sim
 \Norm(0, \hat \Pi_{1n} \Sigma)$ and separately controlling the bounded Lipschitz distance
 (i) between the  laws of $(Y_{1n}, \hat \Pi_{1n}, Y_{2n}), (Y_{1n}, \hat \Pi_{1n}, Y_
 {2}^*)$ and (ii) between the laws of $(Y_{1n}, \hat \Pi_{1n}, Y_{2}^*), (Y_1, \Pi_1,
 Y_2)$. (i) is manageable by conditioning on $Y_{1n}, \hat\Pi_{1n}$ and applying a
 conditional central limit theorem for $Y_{2n}$, and (ii) is manageable through the
 convergence of $(Y_{1n}, \hat \Pi_{1n})$. This decomposition is essentially the same as
 the decomposition in Appendix J.3 of \citet{niu2025assumption}.
\end{rmksq}

 Notably, the asymptotics does not require truncation of $\Pi_{2:T-1,n}$ and instead
 solely requires a pruning step for $\Pi_{Tn}$, unlike existing work
\citep{zhang2020inference,niu2025assumption}. This is because the statistic \eqref
 {eq:real_wls_coef} can be written in terms of the scaled means $Y_t \propto \hat\Pi_{tn}
 (X_ {tn} - \mu)$ instead of $X_ {tkn}.$ The statistics
$Y_{tkn}$ admit a central limit theorem uniformly over $\Pi_{tkn} \in [0,1]$.

\section{General conditional inference} %
\label{sec:inference_in_the_gaussian_model}

If we have additional knowledge of the allocation algorithm as well as the inferential
target and stopping time, then we can design optimal conditional inference procedures that
condition on less information and are hence  more powerful. 

Let us first introduce a general recipe for constructing optimal
conditional procedures by analyzing the likelihood of the observed data. Suppose we wish
to condition on growing information sets $\mathcal F_{T-1} (X_ {1:T-1})$ such that
the experimental design $(\one(T\le t),
\Pi_{1:t}, \eta(X_1,\ldots, X_t; t))$ is measurable with respect to $\mathcal
F_{t-1}$. The law of the observed data, for a given stopping time $T'$, is an exponential
family:
\begin{align*}
p_\mu(X_{1:T} \mid T=T', \mathcal F_{T'-1}) &= \frac{1}{p(\mathcal F_{T'-1})}\prod_{t=1}^
{T'}
\frac{1}{\sqrt{(2\pi)^k \det(V_t)}} \exp\pr{ -\frac{n}{2}(X_t - \mu)' V_t^{-1} (X_t - \mu)
} \\ 
&= \exp\pr{
    n\mu' \sum_{t=1}^{T'} V_t^{-1} X_t
} h(X_{1:T'}) g(\mathcal F_{T'-1}, \mu) \numberthis \label{eq:likelihood}
\end{align*}
where $h(\cdot)$ does not depend on $\mu$. Thus, a sufficient statistic for $\mu$ with
respect to the law of $\pr{X_{1:T} \mid T, \mathcal F_{T-1}}$ is \[
S = n \sum_{t=1}^{T} V_t^{-1} X_{t}. 
\]
For any $\tau = \eta'\mu$, where without loss of generality $\norm{\eta} = 1$, given an
orthogonal matrix $Q$ whose first row is $\eta$ we can further partition \[S'\mu =
S'Q'Q\mu = (\eta'S)
(\eta'\mu) + (\eta_\perp S)'(\eta_\perp\mu) \equiv U\tau + U_\perp'\tau_\perp.\]

By the results in section 5.5 of \citet{pfanzagl2011parametric}, optimal inferences for $\tau$,
conditional on $\mathcal F_{T-1}$,\footnote{\citet{adusumilli2023optimal} studies
Neyman--Pearson-style testing with the likelihood \eqref{eq:likelihood}.} are based on the distribution of
$U
\mid
U_\perp,
\mathcal F_{T-1}$, which depends solely on $\tau$ and is stochastically increasing in
$\tau$. For instance, the optimal $\alpha$-quantile unbiased estimator for $\tau$, based
on observations
$(u, u_\perp, \mathcal F_{T-1}(X_{1:T-1})),$ can be computed by $\hat\tau_\alpha$ such that \[
\P_{\hat\tau_\alpha}\pr{U \le u \mid U_\perp=u_\perp, \mathcal F_{T-1}(X_{1:T-1})} =
\alpha.
\]
The optimal $(1-\alpha)$ equal-tailed confidence interval for $\tau$ is
$[\hat\tau_{\alpha/2}, \hat\tau_{1-\alpha/2}]$. Indeed, this general recipe based on
the likelihood can also be
used to show that the optimal conditional inference procedure in Thompson sampling corresponds
to the procedure based on \eqref{eq:asymp_variance}.

If choices of inferential target and stopping time can be expressed in
terms of linear inequalities in $X_ {1:T-1}$, then the above recipe is particularly
tractable. Like our results in \cref{sub:leftover}, the optimal conditional procedures
dominate using only the last batch.

\begin{rmksq}[Parametric assumptions]
\label{rmk:expofam}
In principle, this recipe for generating conditional inference procedures is not specific
to Gaussian likelihood for $X_{t} \mid \mathcal F_{t-1}$. One could similarly allow for
cases where $X_{t} \mid \mathcal F_{t-1}$ is an exponential family \[
  p(X_{t} \mid \mathcal F_{t-1}) \propto \exp\pr{\mu' S(\mathcal F_t) + \zeta(\mathcal
  F_{t-1})' S_2(X_t)}
\](The Gaussian case corresponds to $S(\mathcal F_t) = nV_t^{-1} X_t$ and $\zeta
 (\mathcal F_{t-1})' S_2(X_t) = X_t'V_t^{-1} X_t$, which is linear in quadratic functions
 of $X_t$.) Conditioning on $\zeta (\mathcal F_
 {t-1})$ ensures that the statistical model has an exponential family structure rather
 than a curved exponential family structure.
  The exponential family structure is then exploited for optimal inference (e.g., optimal
  quantile-unbiased estimators or uniformly most powerful unbiased tests). However, a
  non-Gaussian exponential family structure for batchwise statistics would be hard to
  motivate in practice, since practitioners are typically interested in batchwise means. 
\end{rmksq}

\begin{rmksq}[Sufficient statistics] 

Notably, this procedure is based on $(U, U_\perp)$,
which are sufficient for $\mu$ with respect to the law of $X_{1:T} \mid  \mathcal F_{T-1},
T$. In contrast, we notice that the procedure in \citet{zhang2020inference} is not based
on the sufficient statistic. \Cref{asub:zhang_bet_proof} shows that it can be written as a
linear function of the sufficient statistics and an independent Gaussian noise term. The
natural estimator for $\tau$ from \citet{zhang2020inference} (the mid point of the
confidence interval) is thus dominated by its conditional expectation given the sufficient
statistics (i.e. Rao--Blackwellization). Moreover, we observe that the confidence
intervals in \citet{zhang2020inference} is also not bet-proof
\citep{muller2016credibility}, in the sense that there exists an event of the data with
positive probability on
which the confidence interval undercovers uniformly over $\mu \in \R^K$.
\end{rmksq}

\subsection{Conditional inference for polyhedral algorithms} %
\label{sub:conditional_inference_in_polyhedral_algorithms}

We consider a class of algorithms such that the support of $\mathcal F_
{T-1} \equiv (\Pi_{2:T}, T,
\eta)$ is finite. For every $(\pi, t_0, h)$ in its support, there exist a
conformable matrix $A(\mathcal F_{T-1})$ and vector $b(\mathcal F_{T-1})$ such that for
all $\mu$, \begin{align*}
&\P_\mu\bk{ (\Pi_{2:T}, T,
\eta) = (\pi, t_0, h) \cap A(\mathcal F_{t_0-1}) \bX_{1:t} \le b(\mathcal F_{t_0-1})}
\\
&=\P_\mu
\bk{ (\Pi_{2:T}, T,
\eta) = (\pi, t_0, h)} =\P_\mu\bk{ A(\mathcal F_{t_0-1}) \bX_{1:t_0} \le b(\mathcal F_
{t_0-1})}.
\end{align*}
As an example, suppose $T, \eta$ are fixed, but at each batch $t$, the assignment
algorithm $\epsilon$-greedy
(\cref{ex:egreedy}) favors the arm $k$ with the highest cumulative arm mean
$W_{(t-1),k}$. Each realization of $\pi_{2:T}$ then corresponds to a sequence of
$\epsilon$-greedy winners $(k_1,\ldots, k_{T-1})$. Thus, up to a measure-zero event of
ties, the sequence of winners $(k_1,\ldots, k_{T-1})$ is equivalent to the inequalities \[
\br{W_{t,k_t} \ge W_{t, \ell} : t = 1,\ldots, T-1; \ell = 1,\ldots, K}.
\]
Since $W_{t, \ell}$ is linear in $X_{1},\ldots, X_t$, we may represent these inequalities
as the polyhedron $A(\mathcal F_{T-1}) \bX_{1:T} \le b(\mathcal F_{T-1})$.

Following our conditional inference strategy, optimal inference for $\tau = \eta'\mu$
depends on the distribution \[ (U \mid U_\perp, \mathcal F_{T-1}) \sim (U \mid
U_\perp,
A(\mathcal F_{T-1}) \bX_{1:T} \le b(\mathcal F_{T-1}))
\]
where $U = \eta'S$, $U_\perp = \eta_\perp S$, and $S_k = \sum_{t=1}^T \frac{n_t \Pi_{tk}}
{\sigma_k^2} X_{tk}$. The following theorem makes explicit this distribution, and, in
particular, its dependence on $\mu$ only through $\tau$. 

\begin{restatable}{theorem}{thmpolyhedral}
\label{thm:polyhedral}
    Assume, without loss of generality, that $\norm{\eta} = 1$. Define $R = n\sum_{t=1}^T
    V_t^
    {-1}$. Let $G = 1_T' \otimes I_K = [I_K, \ldots, I_K]$ be a $K \times TK$ matrix.
    Let $G_\perp$ be a $(T-1)K \times TK$ matrix whose rows are unit vectors that are
    orthogonal to each other and to the rows of $G$.
    Let $\bV = \frac{1}{n}\diag(V_1, \ldots, V_T)$
    be a $TK \times TK$ diagonal matrix. Let $\tilde G, \tilde G_\perp$ partition the
    inverse of the following matrix: \[
\colvecb{2}{G \bV^{-1}}{G_\perp}^{-1} = \bk{\tilde G, \tilde G_\perp}.
    \] where $\tilde G$ is $TK \times K$. Let $ c = \frac{1}{\eta' R^
     {-1} \eta}. $ Then \[
\pr{\frac{U}{c} \mid U_\perp=u_\perp, A(\mathcal F_{T-1}) \bX_{1:T} \le b(\mathcal F_{T-1})} \sim \pr{Z_1 \mid M(\mathcal F_{T-1}) Z \le m(\mathcal F_{T-1}) }
    \]
    where the right-hand side follows the law \[Z = (Z_1, Z_2')' \sim \Norm\pr{
\colvecb{2}{\tau + k'u_\perp}{0}, \begin{bmatrix}
    \frac{1}{c} & 0 \\
    0 & G_\perp \bV G_\perp'
\end{bmatrix}
    } \quad k' = \frac{1}{c}\eta' G \bV^{-1} G' \eta_\perp' (\eta_\perp G\bV^{-1}G'
    \eta_\perp')^
    {-1}\]
The constraints are defined by \[
M(\mathcal F_{T-1}) = \bk{cA(\mathcal F_{T-1}) \tilde G \eta, A(\mathcal F_{T-1}) \tilde
G_\perp} \quad m(\mathcal F_{T-1})  = b(\mathcal F_{T-1})-A(\mathcal F_{T-1}) \tilde G 
(I-\eta\eta') S.
\]
\end{restatable}

At any realization of the data $X_{1:T}$, the constraints and parameters \[(m(\mathcal F_
{T-1}) , M(\mathcal F_{T-1}), c, G_\perp \bV G_\perp')\] are functions of known
quantities. The law $(Z_1 \mid M(\mathcal F_ {T-1}) Z \le m(\mathcal F_{T-1}) )$ depends
only on the parameter of interest $\tau$. 

For testing the point null hypothesis $H_0: \tau = \tau_0$, we can draw from the
corresponding distribution $(Z_1 \mid M(\mathcal F_ {T-1}) Z \le m (\mathcal F_{T-1}) )$
induced by the null value $\tau_0$. The observed statistic is $z_1 = U/c$. Valid tests are
then constructed by comparing $z_1$ to its distribution under the null. A
level-$(1-\alpha)$ equal-tailed test, for instance, rejects $H_0$ when $z_1$ falls in the
lower $\alpha/2$ or the upper $\alpha/2$ quantiles of the distribution $(Z_1 \mid
M(\mathcal F_ {T-1}) Z \le m (\mathcal F_{T-1}) )$ under $\tau_0$.

In terms of computation, we can efficiently draw from $(Z \mid M(\mathcal F_{T-1}) Z \le
m(\mathcal F_{T-1}) )$ via Gibbs sampling \citep{taylor2016restrictedmvn},\footnote{We
thank \blind{[anonymous] }for this suggestion.} since the
conditional distribution of each coordinate of $Z$ is a truncated Gaussian. Confidence
intervals can be constructed by inverting tests of $H_0 : \tau = \tau_0$ for a range of
$\tau_0$. Computation of confidence intervals does not require drawing the distribution of
$Z$ for every candidate value of $\tau_0$. Note that, since the likelihood ratio between
two candidate values $(\tau_1, \tau_0)$ is known,
\[
\frac{p_{\tau_1}(z_1 \mid M(\mathcal F_{T-1}) Z \le m(\mathcal F_{T-1}) )}{p_{\tau_0}(z_1 \mid M(\mathcal F_{T-1}) Z \le m(\mathcal F_{T-1}) )} = \exp\pr{
    c(\tau_1 - \tau_0) z_1
},
\]
samples from $p_{\tau_0}$ can be reweighted to obtain estimators for quantities under $p_
{\tau_1}$. As a result, inference based on \cref{thm:polyhedral} is computationally
efficient, at least for moderately-sized $(T,K)$.

\section{Simulation evidence} %
\label{sec:simulation_evidence}

\begin{table}[tbh]
\begin{center}
Fixed Target 

(inference on $\mu_3$)

\begin{tabular}{lrrr}
\toprule
 & Rejection rate & Average length & Median length relative to last \\
\midrule
Leftover & 0.052 & 0.582 & 0.890 \\
ZJM & 0.049 & 0.292 & 0.489 \\
Last-only & 0.048 & 0.690 & 1.000 \\
Anytime-valid & 0.004 & 0.692 & 0.954 \\
\bottomrule
\end{tabular}
\begin{proof}[Notes]
10,000 replications. Nominal 5\% test. ``Median length relative to last'' takes the
median of the ratio of confidence interval lengths relative to last-only, as opposed to
the ratio of median lengths.
\end{proof}
\end{center}

\vspace{1em}

\begin{center}
Adaptive Target 

(inference on arm with highest sample mean in the first $T-1$ batches)

\begin{tabular}{lrrr}
\toprule
 & Rejection rate & Average length & Median length relative to last \\
\midrule
Leftover & 0.050 & 0.316 & 0.917 \\
ZJM & 0.071 & 0.203 & 0.569 \\
Last-only & 0.049 & 0.353 & 1.000 \\
Anytime-valid & 0.006 & 0.418 & 1.167 \\
\bottomrule
\end{tabular}

\begin{proof}[Notes]
10,000 replications. Nominal 5\% test. ``Median length relative to last'' takes the
median of the ratio of confidence interval lengths relative to last-only, as opposed to
the ratio of median lengths.
\end{proof}
\end{center}
\caption{Thompson sampling experiment}
\label{tab:ts}
\end{table}

We consider two simulated experiments. Both experiments have $T = 4, K = 3, \mu =
[0,0,0]'$, $n_1 = n_2 = n_3 = n_4 = 200$.  To demonstrate asymptotic validity of our
inference procedures, we take the individual outcomes to be i.i.d. Rademacher random
variables and we plug in estimated versions of $\Sigma$ in the following empirical
exercises. We consider a Thompson sampling experiment and an $\varepsilon$-greedy
experiment.\footnote{The Thompson sampling algorithm \emph{prunes} the last-batch assignment
probabilities at $0.01$ (See \cref{ex:pruned_ts}), which is needed for our asymptotic
results in \cref{sub:asymptotics}. The $\varepsilon$-greedy algorithm
chooses
$\varepsilon = 0.1$. }

\begin{figure}[tb]
    \centering
    \includegraphics{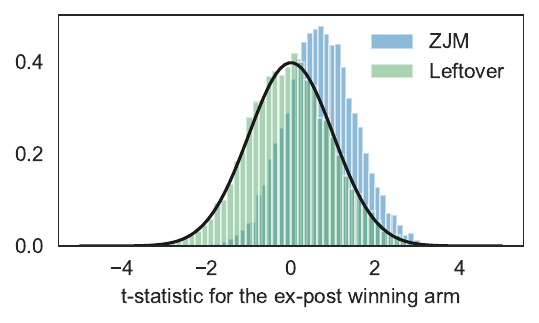}
    \caption{Distribution of $t$-statistics for adaptive target in the Thompson sampling
    experiment}
    \label{fig:t_stat}
\end{figure}

\begin{table}
\begin{center}
Fixed Target 

(inference on $\mu_3$)

\begin{tabular}{lrrr}
\toprule
 & Rejection rate & Average length & Median length relative to last \\
\midrule
Leftover & 0.051 & 0.561 & 0.888 \\
ZJM & 0.053 & 0.286 & 0.430 \\
Last-only & 0.056 & 0.690 & 1.000 \\
Polyhedral & 0.053 & 0.395 & 0.622 \\
Anytime-valid & 0.004 & 0.493 & 0.663 \\
\bottomrule
\end{tabular}
\end{center}
\vspace{1em}

\begin{center}
Adaptive Target 

(inference on arm with highest sample mean in the first $T-1$ batches)

\begin{tabular}{lrrr}
\toprule
 & Rejection rate & Average length & Median length relative to last \\
\midrule
Leftover & 0.050 & 0.279 & 0.951 \\
ZJM & 0.070 & 0.203 & 0.655 \\
Last-only & 0.051 & 0.310 & 1.000 \\
Polyhedral & 0.050 & 0.246 & 0.792 \\
Anytime-valid & 0.005 & 0.440 & 1.365 \\
\bottomrule
\end{tabular}
\end{center}

\begin{proof}[Notes]
10,000 replications. Nominal 5\% test. ``Median length relative to last'' takes the
median of the ratio of confidence interval lengths relative to last-only, as opposed to
the ratio of median lengths.
\end{proof}
\caption{$\varepsilon$-greedy experiment}
\label{tab:eg}
\end{table}

\begin{figure}[tb]
    \centering
    \includegraphics{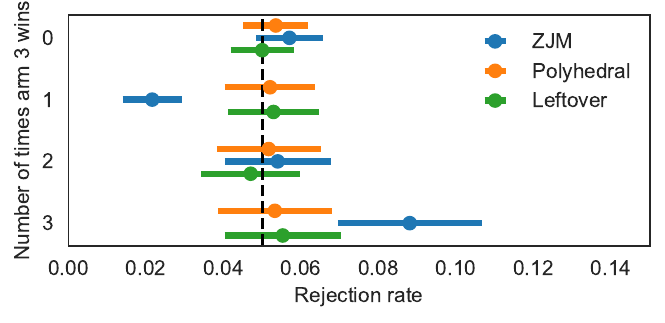}

\begin{proof}[Notes]
    Simultaneous 95\% confidence intervals for coverage shown.
\end{proof}

    \caption{Conditional coverage in $\varepsilon$-greedy experiment}
    \label{fig:conditional_coverage_eg}
\end{figure}

For Thompson sampling, \cref{tab:ts} displays coverage and length of different confidence
intervals. In \cref{tab:ts}, 
\begin{itemize}
  \item Leftover is the
procedure in \cref{sub:leftover} (\cref{thm:maintext})
\item ZJM is the procedure in
\citet{zhang2020inference},\footnote{That is, we use the estimated counterpart of the
following pivotal quantity \[ Z_{\text{ZJM}} = \sum_{t=1}^T V_{t}^{-1/2} (X_t - \mu) \sim
\Norm(0, T I_K).
\]}
\item Last-only is the procedure that only uses the last batch
\item Anytime-valid applies the procedure in \citet{waudby2024estimating} that treats the
 inverse propensity-weighted outcome $Y_{ik}^* \equiv D_{ik} Y_{i} /
\Pi_{tk}$, for a unit $i$ in batch $t$, as a sequence of bounded random variables with
 mean $\mu_k$.\footnote{\citet{waudby2024estimating} provides a procedure that takes in a
 sequence of bounded random variables with mean $\mu$ and produces a confidence sequence
 for $\mu$. Applying this procedure to the sequence $Y_{ik}^* \equiv D_{ik} Y_{i} /
\Pi_{tk}$, $i=1,\ldots,n$, yields a confidence sequence for $\mu_k$. We use the default
 implementation of \texttt{betting\_cs} in the
 \texttt{confseq} Python package \citep{confseq}.} 
\end{itemize} Among these procedures, ZJM and Anytime-valid are not designed to control
 Type I error for adaptive targets. The confidence intervals are truncated to the
 parameter space $[-1,1]$. 

We compare the
procedures over two setups, fixed target and adaptive target. In terms of rejection rates,
the conditional procedures control size for both setups, and ZJM does not control size for
the adaptive setup (\cref{fig:t_stat} plots the distribution of $t$-statistics in the
adaptive setup). Anytime-valid inference \citep{waudby2024estimating} appears to
underreject for both regimes. 

 In terms of length, the improvement in median length from using the leftover information
 is about 10\% for this particular setup of the experiment, relative to only using the
 last batch, whereas the ZJM interval is 40--50\% shorter than only using the last batch.
 This confirms that the improvement that the additional statistic $L$ provides relative
 to last-batch-only is mild but non-trivial (a 10\% improvement in length is comparable
 to a 20\% increase in sample size). On average, we do sacrifice confidence interval
 length to maintain conditional validity, relative to the unconditionally valid procedure
 ZJM. On the other hand, since our results are asymptotic and leverage the batch
 structure, they generally yield shorter intervals than finite-sample inference results
 from the anytime-valid inference literature. 

Similarly, we additionally compare the behavior of the Polyhedral inference procedure
(\cref{thm:polyhedral}) for an $\varepsilon$-greedy experiment in \cref{tab:eg}. For
inference on a fixed target, all procedures achieve their nominal size. The Polyhedral
procedure generates 40\% shorter confidence intervals  (measured in terms of median
length) than Last-only, which is more than double the improvement of Leftover. For
inference on the adaptive target, we again see that the conditional procedures maintain
nominal size (at the cost of greater length), whereas ZJM does not.

Lastly, we plot conditional behavior of these procedures in
\cref{fig:conditional_coverage_eg}, where we condition on the number of times the
inference target, arm 3, is the $\varepsilon$-greedy winner. We find that the conditional
procedures indeed control conditional size, whereas ZJM over-rejects when arm 3 (for whose
mean we perform inference) wins most and least often, and compensates for the
over-rejection by under-rejecting on other sequences of $\varepsilon$-greedy winners.

\section{Conclusion} %
\label{sec:conclusion}

This paper investigates inference conditional on the experimental design in batched
adaptive experiments, and explores the potential of improving upon inference procedures
that only use the last batch. For location-invariant experimental designs, we find there
is a scalar statistic, beyond the last-batch result, which is left over after
conditioning.
Using this additional statistic provides a free lunch improvement for statistical
inference, relative to using the last batch alone. For polyhedral experimental designs, we
characterize optimal conditional inference procedures and demonstrate their computational
tractability.

\bibliographystyle{aer}
\bibliography{main.bib}

\appendix

\numberwithin{equation}{section}

\newpage

\begin{center}
\textbf{Appendix for ``Optimal Conditional Inference in Adaptive Experiments''}

\today
\end{center}

\DoToC

\vspace{2em}

\section{Proofs for results in the main text}
\label{asec:proofs_maintext}
\subsection{Improvability}

\lemmaimprovability*

\begin{proof}
    Suppose not. Then there exists some event $E$ and some $\eta_0$ where $\P(X_{1:T-1}
    \in E) > 0 $ and \[
        \P\pr{
            \eta_0'\mu \in \hat C_{\eta_0}(\alpha) \mid X_{1:T-1} \in E
        } < 1-\alpha.
    \]
    Let $\eta = \eta_0$ on $E$ and $\eta = [1, 0,\ldots, 0]$ on $E^C$. Then \[
        \P\pr{
            \eta'\mu \in \hat C_{\eta}(\alpha) \mid \eta(X_{1:T-1}) = \eta_0
        } < 1-\alpha
        \]
    and $\P(\eta(X_{1:T-1}) = \eta_0) > 0$. This contradicts
    \eqref{eq:conditional_requirement}.
\end{proof}

\subsection{Leftover information $L$}
\label{sub:leftover-proof}

\thmleftover*

\begin{proof}
We first scale $X_t$ by $\Pi_t$ and use \cref{lemma:leftover}---doing so mainly conforms
with our subsequent asymptotic results. Recall that $\Pi_t = \diag(\Pi_{t1},\ldots,
\Pi_{tK})$, $\Sigma = \diag (\sigma_1^2,\ldots, \sigma_K^2)$, $V_T =
\frac{1}{n}\frac{1}{c_T}\Pi_T^{-1}
\Sigma$ where $c_t = n_t / n$. Note that we can represent conditioning on $\Delta X_
{1:T-1}$ as  conditioning on certain linear transformations $
A \bX_{1:T-1}$, where $\bX_{1:T-1} = [X_1',\ldots, X_{T-1}']'$. 

To apply \cref{lemma:leftover}, define $Y_t = \sqrt{c_t} \Pi_t \sqrt{n} (X_t - \mu)$ and
let $\bY_{1:t} = [Y_1', \ldots, Y_t']'$. When $\Pi_t > 0$, it suffices to study the joint
distribution of $L$ and $Y_T$ conditional on $B \bY_{1:(T-1)}$, where \[ B = A
\begin{bmatrix}
    \sqrt{nc_1} \Pi_1 &  \\
     & \ddots & \\
     &  & \sqrt{nc_{T-1}} \Pi_{T-1}
\end{bmatrix}^{-1}.
\]
Note that since $A$ takes pairwise differences, $A1_{(T-1)K} = 0$. Thus, \[B\colvecb{3}{
\sqrt{c_t}\Pi_11_K}{\vdots}{\sqrt{c_{T-1}} \Pi_{T-1}1_K} = 0.\]

Note that we can define $A_t, B_t$ to be the submatrices of $A, B$ that only act on $Y_
{1:t}$ such
that $A=A_{T-1}$, $B=B_{T-1}$. Note that for all $t,$ \eqref{eq:B_condition} is satisfied.
Thus we can apply \cref{lemma:leftover} to show that \[
\colvecb{2}{\sqrt{n}(L-\lambda'\mu)}{Y_T} \mid T, B_{T-1} \bY_{1:T-1} \sim \Norm\pr{0,
\begin{bmatrix}
    \lambda'1 & 0 \\
    0 & \Pi_T \Sigma
\end{bmatrix}}
\]
Therefore, \[
\colvecb{2}{L}{X_T} \mid T, A \bX_{1:T-1} \sim \Norm\pr{\colvecb{2}{\lambda'\mu}{\mu},
\begin{bmatrix}
    \lambda'1/n & 0 \\
    0 & \frac{1}{n} V_T
\end{bmatrix}}
\]
Since the right-hand side of the above display depends only on $\Pi_{2:T}, T$, \[
\colvecb{2}{L}{X_T} \mid T, A \bX_{:T-1} \sim \colvecb{2}{L}{X_T} \mid \Pi_{2:T}, T.
\]
This proves  \eqref{eq:dist_suff_stat_diff}.

To prove the ``moreover'' part, we show that the information in $\bX_{1:T-1}$ can be
recovered from $L$ and $\Delta X_{1:T-1}$. Since $\Delta X_{1:T-1}$ and $L$ are both
linear
transformations of $\bX_{1:T-1}$, it suffices to show that the coefficients that generate
$\Delta X_{1:T-1}$ and $L$ are linearly independent. Precisely speaking, 
note that by construction, $L = \ell'
\bX_{1:T-1}$ where $AD\ell = 0$ (almost surely for some positive definite diagonal matrix
$D$ that depends on $T$). It suffices to show that the rows of $A$ and $\ell $ span the
whole space $\R^{K(T-1)}$. Since $A$ is rank $K (T-1)-1$, it suffices to show that $\ell$
is not spanned by the rows of $A$.

The condition $AD\ell = 0$ is sufficient for this. This shows that $\ell$ is linearly
independent with rows of $A$, since $D\ell \in \mathrm{Nul}(A) = \mathrm{Col}(A')^\perp$.
Since $\ell' D\ell > 0$, $\ell \not\in \pr{\mathrm{Col}(A')^\perp}^\perp =
\mathrm{Col}(A')$. Hence $\ell$ is not in the row space of $A$.  Since $A$ is
rank-deficient by a single rank, we conclude that we can invert $L, A \bX_{1:T-1}$ into
$\bX_{1:T-1}$. 

Hence $(L, X_T, A\bX_{1:(1-T)})$ maps into $X_{1:T} $, which is trivially sufficient for
$\mu$ with
respect to $X_{1:T} \mid T$. Since $(L, X_T, A\bX_{1:(1-T)})$ is sufficient for $\mu$ with
respect to $X_{1:T} \mid T$, $ (L, X_T)$ is sufficient with respect to $X_{1:T} \mid T,
\Delta X_{1: T-1}
\sim X_{1:T} \mid \Delta X_{1:T}$. This proves the ``moreover'' part.
\end{proof}

\begin{lemma}
\label{lemma:leftover}
In the Gaussian model (where we allow $\Pi_{tk} = 0$ with positive probability), recall
that $\Pi_t = \diag(\Pi_{t1},\ldots, \Pi_{tK})$ and $\Sigma =
\diag
(\sigma_1^2,\ldots, \sigma_K^2).$ Let $Y_t =
\sqrt{c_t}\Pi_{t}\sqrt{n}\pr{X_t - \mu}$ such that $
Y_t \mid \Pi_t \sim \Norm(0, \Pi_t \Sigma).
$
Let $\bY_{1:t} = [Y_1', \ldots, Y_t']'$.
Let $B_t = B_t(\Pi_{1:t})$ be a matrix such that \[
B_t \colvecb{3}{\sqrt{c_1}\Pi_11_K}{\vdots}{\sqrt{c_t} \Pi_t1_K} = 0 \numberthis
\label{eq:B_condition}
\]
almost surely.
Assume that, for all $t$, $\one(T > t+1) $ and $ \Pi_{t+1}$ are measurable with respect to
$B_1 \bY_{1:1}, \ldots, B_t \bY_ {1:t}$. Consider $
\sqrt{n}(L-\lambda'\mu) =  1_K' \Sigma^{-1} \sum_{t=1}^{T-1}  \sqrt{c_t} Y_t.
$ Then \[
\colvecb{2}{\sqrt{n}(L-\lambda'\mu)}{Y_T} \mid T, B_1\bY_{1:1}, \ldots, B_{T-1}\bY_{1:T-1}
\sim \Norm
\pr{0,
\begin{bmatrix}
    1'\pr{\Sigma^{-1} \sum_{t=1}^{T-1} c_t \Pi_t}1 & 0 \\
    0 & \Pi_T \Sigma
\end{bmatrix}}.
\]
\end{lemma}

\begin{proof}

We show the following claim with a fixed stopping time: For all $s = 2,\ldots, T_0$, let 
\[ Z_s = 1_K' \Sigma^{-1} \sum_{t=1}^{s-1} \sqrt{c_t} Y_t.
\]
Then, for all $s$, we claim that \[
(Z_s, Y_{s}) \mid B_1 \bY_{1:1},\ldots, B_{s-1} \bY_{1:s-1} \sim \Norm\pr{0, \diag\pr{
    1'\pr{\Sigma^{-1} \sum_{t=1}^{s-1} c_t \Pi_t}1, \Pi_s \Sigma
}}. \numberthis \label{eq:claim}
\]

Given \eqref{eq:claim}, note that for all $s$, \[
\br{(Z_T, Y_T) \mid T=s, B_1 \bY_{1:1},\ldots, B_{T-1} \bY_{1:T-1}} \sim \br{(Z_s, Y_s) \mid
B_1
\bY_
{1:1},\ldots, B_{s-1} \bY_{1:s-1}
},\]
by our assumptions about the stopping time $T$. This proves this lemma. 

To show \eqref{eq:claim}, we induct on $t$. To that end, define \[%
\lambda_t = \colvecb{3}{
\sqrt{c_1}\Sigma^{-1}1_K}{\vdots}{\sqrt{c_t}\Sigma^{-1}1_K} \equiv \colvecb{3}{\tilde
\lambda_1}{\vdots}{\tilde \lambda_t}
\]
We wish to verify the following claims for all $t = 1,\ldots, T-1$:
\begin{enumerate}[label=(\arabic*)]
    \item $\bY_{1:t+1} \mid \br{B_s\bY_{1:s}: s \le t} \sim \Norm\pr{\mu_{t+1 \mid t},
    \Sigma_
    {t+1 \mid t}}$ where \begin{align*}
\Sigma_{t+1 \mid t} &= \begin{bmatrix}
    \Sigma_{t|t} & 0 \\
    0 & \Sigma \Pi_{t+1}
\end{bmatrix} \quad \Sigma_{t|t} = \Sigma_{t|t-1} - \Sigma_{t|t-1} B_t'(B_t \Sigma_
{t|t-1} B_t')^{+} B_t \Sigma_{t|t-1},
\\\mu_{t+1|t} &= \colvecb{2}{\mu_{t|t}}{0} \quad \mu_{t|t} = \mu_{t|t-1} + \Sigma_{t|t-1}
B_t'(B_t \Sigma_
{t|t-1} B_t')^{+} B_t (\bY_{1:t} - \mu_{t|t-1}) \\
\Sigma_{1|0} &= \Sigma\Pi_1\\
\mu_{1|0} &= 0
    \end{align*}

    \item
    \[\Sigma_{t|t}\lambda_{t} = \Sigma_{t|t-1} \lambda_{t} = \colvecb{3}{\sqrt{c_1}
    \Pi_11_K}{\vdots}{
    \sqrt{c_{t}} \Pi_{t} 1_K}\]
    \item $B_{t} \Sigma_{t|t-1}\lambda_{t} = 0$
    \item $\lambda_t'\mu_{t|t} = 0$
\end{enumerate}

Assuming items (1)--(4), we show \eqref{eq:claim}. Note that $
Z_s = \lambda_{s-1}'\bY_{1:s-1}.
$
Hence,
\[
\colvecb{2}{Z_s}{Y_s} = \begin{bmatrix}
    \lambda_{s-1}' & 0\\
    0 & I_K
\end{bmatrix} \bY_{1:s}
\]
and they are jointly Gaussian.
Its mean conditional on $\br{B_t\bY_{1:t} : t \le s-1}$ is thus \[
\begin{bmatrix}
    \lambda_{s-1}' & 0\\
    0 & I_K
\end{bmatrix} \colvecb{2}{\mu_{s-1|s-1}}{0} = 0
\]
and its conditional variance is  \begin{align*}
&\begin{bmatrix}
    \lambda_{s-1}' & 0\\
    0 & I_K
\end{bmatrix} \begin{bmatrix}
    \Sigma_{s-1|s-1} & 0 \\
    0 & \Sigma \Pi_{s}
\end{bmatrix} \begin{bmatrix}
    \lambda_{s-1}' & 0\\
    0 & I_K
\end{bmatrix}' \\
&= \begin{bmatrix}
    \lambda_{s-1}'\Sigma_{s-1|s-1}\lambda_{s-1} & 0 \\ 0 & \Sigma \Pi_s
\end{bmatrix} = \begin{bmatrix}
    1_K'\pr{\Sigma^{-1} \sum_{t=1}^{s-1} c_t \Pi_t}1_K & 0 \\ 0 & \Sigma \Pi_s
\end{bmatrix}
\end{align*}
as desired. Thus, items (1)--(4) prove \eqref{eq:claim}. 

Finally, to verify items (1)--(4), consider the base case $t=1$. Then $\bY_1 = Y_1 \sim
\Norm(0,
\Sigma\Pi_1) = \Norm(\mu_{1|0}, \Sigma_{1|0})$. Note that \[
Y_1 \mid B_1 Y_1 \sim \Norm\pr{\underbrace{\mu_{1|0} - \Sigma_{1|0}B_1'(B_1\Sigma_
{1|0}B_1')^{+}B_1
(Y_1 - \mu_{1|0})}_{\mu_{1|1}}, \underbrace{\Sigma_{1|0} - \Sigma_{1|0}B_1'(B_1\Sigma_
{1|0}B_1')^
{+}B_1\Sigma_
{1|0}}_{\Sigma_{1|1}}}
\]
and $(Y_2 \mid B_1 Y_1) \sim (Y_2 \mid \Pi_2) \sim \Norm(0, \Pi_2 \Sigma)$. Thus $\bY_
{1:2} \mid B_1 \bY_{1:1} \sim \Norm(\mu_{2|1}, \Sigma_{2|1})$. This verifies item (1). For
items (2) and (3), note that $\Sigma_{1|0}\lambda_1 = \Sigma\Pi_1\sqrt{c_t} \Sigma^ {-1}
1_K
=
\sqrt{c_t} \Pi_1 1_K$. Note that $B_1 \Sigma_{1|0}\lambda_1 = B_1 \sqrt{c_t} \Pi_1 1_K =
0$ by assumption. Hence $\Sigma_{1|1} \lambda_1 = \Sigma_{1|0}\lambda_1$. This verifies
items (2) and (3). Finally, $\lambda_1 \mu_{1|1} = \lambda'\mu_{1|0} = 0$ by (3). This
verifies item (4).

For the inductive case, assume that items (1)--(4) holds for all $s \le t-1$. Then \[
\bY_{1:t} \mid \br{B_s \bY_{1:s} : s \le t-1} \sim \Norm(\mu_{t|t-1}, \Sigma_{t|t-1})
\]
is jointly Gaussian. Hence \[
\bY_{1:t} \mid \br{B_s \bY_{1:s} : s \le t} \sim \Norm\pr{\mu_{t|t}, \Sigma_{t|t}}.
\]
Note too that \[
Y_{t+1} \mid \br{B_s \bY_{1:s} : s \le t} \sim (Y_{t+1} \mid \Pi_{t+1}) \sim \Norm\pr{0,
\Pi_{t+1} \Sigma}.
\]
independently from $\bY_{1:t}$. This verifies item (1) for $t$.

For items (2) and (3), note that \[
\Sigma_{t|t-1} \lambda_t = \colvecb{2}{\Sigma_{t-1|t-1}\lambda_{t-1}}{\Sigma \Pi_{t} \tilde
\lambda_t} = \colvecb{4}{\sqrt{c_1}\Pi_1 1_K}{\vdots}{\sqrt{c_{t-1}} \Pi_{t-1} 1_K}{
\sqrt{c_t} \Pi_t 1_K}
\text{ and }
B_t \Sigma_{t|t-1} \lambda_t = B_t \colvecb{3}{\sqrt{c_1}\Pi_1 1_K}{\vdots}{
\sqrt{c_t} \Pi_t 1_K} = 0
\]
by assumption. Hence $\Sigma_{t|t-1} \lambda_t = \Sigma_{t|t} \lambda_t$. Finally, for
(4), $\lambda_t'\mu_{t|t} = \lambda_t'\mu_{t|t-1} = \lambda_{t-1}'\mu_ {t-1|t-1} = 0$ by
(3). This completes the proof.
\end{proof}

\subsection{Sufficiency in Thompson sampling}
\label{sub:thompson-sufficiency-proof}

\propovercondition*

\begin{proof}
We will show the following two claims: 
\begin{enumerate}
\item Under Thompson sampling, $(L, \Pi_{2:T})$ is sufficient for $\mu$ with respect
    to $X_
    {1:T} \mid T$.

    This follows directly from \cref{lemma:thompson_sufficiency}.

    \item If $(L, \Pi_{2:T})$ is sufficient for $\mu$ with respect to $X_{1:T-1} \mid T$,
    then $(L, X_T)$ is sufficient for $\mu$ with respect to $X_{1:T} \mid T, \Pi_{2:T}$.

Observe that \begin{align*}
p_\mu(X_{1:T} \mid T) &= p_\mu(X_{1:T-1} \mid T) p_\mu(X_T \mid X_{1:T-1}, T) \\&= g(X_{1:T-1} \mid \Pi_{2:T},
L, T) p_\mu(\Pi_{2:T}, L \mid T) p_\mu(X_T \mid \Pi_{2:T}, T)
\end{align*}
where $g$ does not depend on $\mu$ since $(\Pi_{2:T}, L)$ is sufficient with respect to
$X_{1:T-1} \mid T$. Since $\mu$ enters only through functions that depend on $L, X_T, \Pi_{2:T}$,
we conclude that $(L, X_T, \Pi_{2:T})$ is sufficient for $\mu$ with respect to $X_{1:T} \mid T$.
Thus, $(L, X_T)$ is sufficient for $\mu$ with respect to $X_{1:T} \mid \Pi_{2:T}, T$.
\qedhere
\end{enumerate}
\end{proof}

\begin{lemma}
\label{lemma:thompson_sufficiency}
Under Thompson sampling, assume that $\Xi_t$ is measurable with respect to $\Pi_{2:t}$.
Then $(\Pi_{2:T}, L)$ is sufficient for $\mu$ with respect to $X_
{1:T-1} \mid T$.
\end{lemma}
\begin{proof}
Let $L_s = \sum_{k=1}^K \sum_{t=1}^{s-1} \frac{n_t \Pi_{tk}}{n \sigma_k^2} X_{tk}$. We
shall show the claim that $(\Pi_{2:s}, L_s)$ is sufficient  for $\mu$ with respect to $X_
{1:s-1}$. Given this result (which has a fixed stopping time), we note that, since $
\br{T=s}$ is measurable with respect to
$\Pi_{2:s}$, \[
p(X_{1:T-1} \mid T=s, \Pi_{2:T}, L) = p(X_{1:s-1} \mid T=s, \Pi_{2:s}, L_s) = p(X_{1:s-1}
\mid \Pi_{2:s}, L_s).
\]
Since $(\Pi_{2:s}, L_s)$ is sufficient  for $\mu$ with respect to $X_ {1:s-1}$, the above
display does not depend on $\mu$. As a result, $(\Pi_{2:T}, L)$ is sufficient for $\mu$
with respect to $X_ {1:T-1} \mid T$. This proves \cref{lemma:thompson_sufficiency}
provided the claim holds.

To prove the claim, note that by the \citet{hotz1993conditional} inversion result
\citep{norets2013surjectivity}, $Q$ is a bijection between $\Pi_{t+1}$ and the differences
\[
\nabla_t = (W_{t,1} - W_{t,K}, \ldots, W_{t,K-1} - W_{t, K}).
\]
Thus, it suffices to show that $(\nabla_{1:s-1}, L_s)$ is sufficient for $\mu$ with
respect to $X_{1:s-1}$.

Note that \begin{align*}
L_s &= \sum_{k=1}^K \sum_{t=1}^{s-1} \frac{n_t \Pi_{tk}}{n \sigma_k^2} X_{tk}  \\
&= \sum_{k=1}^K \pr{\sum_{t=1}^{s-1}\frac{n_t\Pi_{tk}}{n\sigma_k^2}}W_{s-1, k} \\
L_s - \sum_{k=1}^{K-1} \pr{\sum_{t=1}^{s-1}\frac{n_t\Pi_{tk}}{n\sigma_k^2}}(W_{s-1, k}-W_
{s-1, K}) &= \pr{\sum_{s=1}^{s-1}\frac{n_t\Pi_{tK}}{n\sigma_K^2}} W_{s-1, K} + \sum_{k=1}^{K-1} \pr{\sum_{t=1}^{s-1}\frac{n_t\Pi_{tk}}{n\sigma_k^2}} W_{s-1, K} \\
&= W_{s-1, K} \sum_{t=1}^{s-1} \sum_{k=1}^K \frac{n_t\Pi_{tk}}{n\sigma_k^2}.
\end{align*}
That is, we can transform $ (\nabla_{1:s-1}, L_s)
$
into $(\Pi_{1:s}, W_{s-1})$, the latter of which is sufficient for $\mu$ with respect to
$X_{1:s-1}$. This completes the proof.
\end{proof}

\subsection{Polyhedral assignment}
\label{sub:polyhedral-proof}

\thmpolyhedral*
\begin{proof}
Note that $S = G\bV^{-1} \bX_{1:T}$ and that $R = G \bV^{-1} G'$. We are interested in the
joint distribution of \begin{align*}
\colvecb{2}{U}{U_\perp} \mid \Pi_{2:T} = \pi &\sim \colvecb{2}{U}{U_\perp} \mid A(\mathcal F_{T-1}) \bX_{1:T} \le
b(\mathcal F_{T-1}) \\
&\sim \colvecb{2}{\eta'G\bV^{-1}}{\eta_\perp' G\bV^{-1}} \bX_{1:T}
\mid A(\mathcal F_{T-1}) \bX_{1:T} \le
b(\mathcal F_{T-1}) \\
&\sim \colvecb{2}{\eta'G\bV^{-1}}{\eta_\perp' G\bV^{-1}} \bX_{1:T}^* \mid A(\mathcal F_{T-1}) \bX_{1:T}^*
\le
b(\mathcal F_{T-1})
\end{align*}
where $\bX_{1:T}^* \sim \Norm(1_T \otimes \mu, \bV)$. Let $S^* = G \bV^{-1} \bX_
{1:T}^*$. Thus, it suffices to study the conditional distribution \[
\underbrace{\eta'G\bV^{-1}  \bX_{1:T}^*}_{c Z_1} \mid \underbrace{\eta_\perp' G\bV^
{-1} \bX_{1:T}^*}_
{Z_3}, A(\mathcal F_{T-1}) \bX_
{1:T}^*
\le b(\mathcal F_{T-1}).
\]
Let $G_\perp \bX_{1:T}^* = Z_2$. 

We first study the Gaussian distribution $(Z_1, Z_2')
\mid Z_3$. This amounts to calculating conditional means and variances. Note that
\begin{align*}
\E[c Z_1 \mid Z_3] &= \eta' R\mu + ck'(Z_3 - \eta_\perp R\mu) \\
\eta'R\mu - ck'\eta_\perp R \mu &= \eta'R\mu - \eta'R \eta_\perp' (\eta_\perp
R \eta_\perp')^{-1} \eta_\perp R\mu\\
&= \eta'R \mu - \eta' R \eta_\perp' (\eta_\perp
R \eta_\perp')^{-1} \eta_\perp R (\eta\eta' + \eta_\perp' \eta_\perp)\mu \\
&= \eta'R\mu - \pr{\eta' R \eta_\perp' (\eta_\perp
R \eta_\perp')^{-1} \eta_\perp R \eta} \cdot \tau - \eta' R \eta'_\perp\eta_\perp\mu \\
&= \eta' R( I - \eta'_\perp\eta_\perp) \mu - \pr{\eta' R \eta_\perp' (\eta_\perp
R \eta_\perp')^{-1} \eta_\perp R \eta} \cdot \tau \\
&= \pr{\eta' R \eta - \eta' R \eta_\perp' (\eta_\perp
R \eta_\perp')^{-1} \eta_\perp R \eta}\tau
\end{align*}
and hence \[
\E[c Z_1 \mid Z_3] =ck' Z_3 + \pr{\eta' R \eta - \eta' R \eta_\perp' (\eta_\perp
R \eta_\perp')^{-1} \eta_\perp R \eta}\tau.
\]

Next, we show that \[
c = \eta' R \eta - \eta' R \eta_\perp' (\eta_\perp
R \eta_\perp')^{-1} \eta_\perp R \eta = \eta' R^{1/2} \pr{ I - R^{1/2} \eta_\perp' (\eta_\perp
R \eta_\perp')^{-1} \eta_\perp R^{1/2}} R^{1/2}\eta
\]
The matrix $\pr{ I - R^{1/2} \eta_\perp' (\eta_\perp
R \eta_\perp')^{-1} \eta_\perp R^{1/2}}$ projects to the space of all vectors
perpendicular to $\eta_\perp R^{1/2}$. Note that this space is one dimensional with basis
vector $R^{-1/2} \eta$. As a result, \[
I - R^{1/2} \eta_\perp' (\eta_\perp
R \eta_\perp')^{-1} \eta_\perp R^{1/2} = \frac{R^{-1/2} \eta \eta' R^{-1/2}}{\eta' R^{-1}
\eta}.
\]
Hence \[
\eta' R^{1/2} \pr{ I - R^{1/2} \eta_\perp' (\eta_\perp
R \eta_\perp')^{-1} \eta_\perp R^{1/2}} R^{1/2}\eta = \frac{1}{\eta' R^{-1}\eta} = c.
\]
Hence \[
\E[Z_1 \mid Z_3] = \tau + k'Z_3.
\]

Note that $\cov(S^*, Z_2) = G\bV^{-1} \bV G_\perp' = 0$ by construction. Hence $\cov
(Z_1, Z_2 \mid Z_3) = 0$,  $\var(Z_2 \mid Z_3) = G_\perp \bV G_\perp'$, and $\E[Z_2
\mid Z_3] = 0$. Lastly, \begin{align*}
\var(cZ_1 \mid Z_3) &= \eta'R\eta - \eta'R\eta'_\perp (\eta_\perp R\eta_\perp')^{-1} \eta'
R \eta = c
\end{align*}
and hence \[
\var(Z_1 \mid Z_3) = 1/c.
\]
This implies that \[
\colvecb{2}{Z_1}{Z_2} \mid Z_3 \sim \Norm\pr{\colvecb{2}{\tau + k' Z_3}{0}, \begin{bmatrix}
    1/c & 0 \\ 
    0 & G_\perp \bV G_\perp'
\end{bmatrix}}.
\]

It remains to show that the constraint $A(\mathcal F_{T-1}) \bX_{1:T}^*
\le
b(\mathcal F_{T-1}) $ is equivalent to $M(\mathcal F_{T-1}) Z \le b(\mathcal F_{T-1})$. This is true by plugging in \[
\bX_{1:T}^* = \tilde G S^* + \tilde G_\perp Z_2 = \tilde G (cZ_1) + \tilde G \underbrace{
(I-\eta\eta')}_{\eta_\perp'\eta_\perp}
S^* + \tilde G_\perp Z_2.
\]
This completes the proof.
\end{proof}

\section{Proofs for results in \cref{sub:asymptotics}}

\subsection{Algorithms satisfying \cref{as:kappa_continuity,as:kappa_pruning}}
\label{sub:kappa-checks}

\begin{exsq}[Pruned Thompson sampling]
\label{ex:pruned_ts}
Pruned Thompson sampling is the assignment algorithm of the form
\eqref{eq:kappa_decomposition}
where $\kappa_{t, (1)} = (\rho \circ Q)(\cdot, \cdot)$ 
and $\kappa_{t, (0)} = Q (\cdot, \cdot)$ for the Gaussian
orthant probability $Q$ defined in
\eqref{eq:orthant} and the pruning function \[
\rho_k(\pi) \propto \pi_k \one(\pi_k \ge \epsilon)
\]
and $\sum_{k=1}^K \rho_k(\pi) = 1$. 

\begin{lemma}
This procedure satisfies \cref{as:kappa_continuity,as:kappa_pruning}.
\end{lemma}
\begin{proof}
It suffices to check that (a) both $\rho \circ Q$ and $Q$ are LIAC for the continuity
statements, (b) $Q_k(v, \Omega) > 0$ when $v_k > 0$. 

Note that (b) is immediately true. Now (a) for $Q$ is immediately true as well, since $Q$
is continuous at every input $\nu, \Omega$ where $\Omega$ is a diagonal matrix with
entries between $\epsilon$ and $1/\epsilon$ and $\nu \in [-\infty, 0]^{K}$. 

Lastly, fix a nonempty $J \subset [K]$. For $\kappa_{t,(1)}$, for a fixed $\Omega$, the
discontinuities
are of the form \[
\bigcup_{k \in J} \br{\nu \in \R^{|J|} : Q_k(\nu^*,\Omega) = \epsilon, \nu_{k}^* = -\infty
\text{ for
$k\not\in J$}, \nu^*_j = \nu_j}.
\]
Note that each entry of the union is measure zero with respect to $\R^{|J|}$. Hence, the
\eqref{eq:adequately_continuous_statement} is satisfied for almost every $\nu$. This
verifies \cref{as:kappa_continuity}. The pruning satisfies \cref{as:kappa_pruning} by
construction.
\end{proof}
\end{exsq}

\begin{exsq}[$\varepsilon$-greedy]
Consider the algorithm satisfying \eqref{eq:simplied_mechanism} and \[
\kappa_{t,(1)}(\nu, \Omega) = \kappa_{t,(0)}(\nu, \Omega) = \begin{cases}
    1-\varepsilon, &\text{ if } \nu_k \ge \max_\ell \nu_\ell \\
    \varepsilon/(K-1) & \text{ otherwise}.
\end{cases}
\]
for some $\varepsilon > 0$ (defined on $\nu$ where the largest entry is unique).

\begin{lemma}
This procedure satisfies \cref{as:kappa_continuity,as:kappa_pruning}.
\end{lemma}
\begin{proof}
For $\epsilon < \varepsilon / (K-1)$, this procedure satisfies \cref{as:kappa_pruning} by
construction. For \cref{as:kappa_continuity}, fix some $J$. Note that the discontinuities
are of the form \[
\bigcup_{k \neq \ell, k\in J, \ell \in J} \br{\nu: \nu_k = \nu_\ell \ge \max_{m\in [K]}
\nu_m > -\infty}.
\]
Each entry of the union is measure zero with respect to $\R^{|J|}$. Hence,
\eqref{eq:adequately_continuous_statement} is satisfied for almost every $\nu$. Moreover,
note that the assignment probabilities are never zero, and so the second part of
\cref{as:kappa_continuity} is automatically satisfied. This verifies
\cref{as:kappa_continuity}.
\end{proof}
\end{exsq}

\subsection{Outline of the general argument}
\label{sub:generic-size-control}

We recall the following notation choices in \cref{sub:asymptotics} and define some more
notation.
\begin{itemize}[wide]
\item Recall that there are at most $T_0 > 0$ batches with an adaptive stopping time
$T_n$.
Assume that the sample sizes $n_t/n \to c_t$ for all $ t= 1,\ldots, T_0$ are fixed
sequences. 
    \item Let $\Pi_{tkn}$ be the algorithm-determined probability of assigning to arm $k$
    in
    batch $t$
    \item Let $N_{tkn} \mid \Pi_{tkn} \sim \Bin(n_t, \Pi_{tkn})$ be the realized number of
    samples assigned to arm $k$ in batch $t$, where $N_{tkn} = \sum_{i\in \mathcal I_t}
    D_{ik}$.
    \item Let $\hat \Pi_{tkn} = N_{tkn}/ n_t$.
 \item Let $X_{tkn} = \frac{1}{N_{tkn} \vee 1} \sum_{i\in \mathcal I_t} D_{ik} X_i^{\obs}$
    be the mean arm value for batch $t$ and arm $k$. If no one is assigned to arm $k$,
    then $X_{tkn} = 0$. Let $X_{tn}$ collect $X_{tkn}$.

   \item  Let $Y_{tkn} = \frac{N_{tkn}}{\sqrt{n_t}} (X_{tkn} - \mu_{nk}) = \sqrt{n_t}
   \hat\Pi_
   {tkn} (X_{tkn} - \mu_{nk})  = \frac{1}{
    \sqrt{n_t}} \sum_{i\in \mathcal I_t} D_{ik} (X_i^\obs - \mu_{nk})$ be a scaled version of
    $X_{tkn}$. Let $Y_{tn}$ collect $Y_{tkn}$.

    \item Let $Y_{:t,k,n} = \sum_{s=1}^t \sqrt{n_s/n} Y_{s,k,n}$

    \item Let $\hat\Pi_{:t,k,n} = \sum_{s=1}^t \hat\Pi_{s,k,n} \frac{n_s}{n}$

    \item Let \[W_{t,k,n} = \frac{
\sum_{s=1}^t N_{skn} X_{skn}}{\sum_{s=1}^t N_{skn}}
    \]
    be the cumulative arm-$k$ mean
    where \[
\sqrt{n}(W_{t,k,n} - \mu_{nk}) = \frac{1}{\hat\Pi_{:t,k,n}} Y_{:t,k,n}.
    \]

    \item Let \eqref{eq:def_sigma2} be an estimate of the arm variance using data up until
   batch $t$. Let $\hat \sigma_{kn}^2 =
   \hat \sigma_{Tkn}^2$.

    \item Let \[
L_n = \sum_{t=1}^{T_n-1} \sum_{k=1}^K \frac{n_t \hat\Pi_{tkn}}{\sqrt{n} \hat \sigma_
{kn}^2}
X_{tkn}
    \] be the properly scaled empirical analogue to $L$.

    \item Let $\lambda_{kn} = \sum_{t=1}^{T-1} \frac{n_t \hat\Pi_{tkn}}{\sqrt{n}
    \hat\sigma_
    {kn}^2}$
    and $\lambda_n = [\lambda_{1n}, \ldots, \lambda_{Kn}]'$.
    Note that \[
L_n - \lambda_n'\mu_{n} = \sum_{t=1}^{T_n-1} \sum_{k=1}^K \frac{n_t \hat\Pi_{tkn}}{
\sqrt{n}
\hat
\sigma_{kn}^2} (X_{tkn} - \mu_{nk})  = \sum_{t=1}^{T_n-1} \sum_{k-1}^K \frac{\sqrt{n_t/n}}
{
\hat \sigma_{kn}^2} Y_
{tkn}.
    \]
\end{itemize}

Our goal is to show that \eqref{eq:real_wls_coef} converges to a Gaussian random variable,
and its limit is conditionally Gaussian given limits of $\Pi_{1:T, n}, T_n, \eta_n$. To do
so, we first analyze the behavior of $(\Pi_{1:T_0, n}, Y_{1:T_0, n})$ without a stopping
time---we could imagine that the experiment carries on after $T_n$. Using the scaling $Y$
instead of $X$ allows for $\Pi_{tkn}$ to be very close to zero or actually zero. 

We show $(\Pi_{1:T_0, n}, Y_{1:T_0, n})$ weakly converges by the following iterative
process. So long as $\Pi_{t+1}$ is a suitably continuous function of the past information
$G_{t,n} = (\Pi_{1:t,n}, Y_ {1:t,n})$, it weakly converges to a function of the limits of
the past information $\Pi_{t+1}(G_t)$. Given this convergence, we also show that $Y_{t+1}$
converges to $\Norm(0, \Pi_{t+1} \Sigma)$, and hence $G_{t+1,n}$ converges (\cref{thm:main_convergence}).

Second, given a stopping time $T_n$ that is a suitably continuous function of the past
information, we can further show that $L_n, Y_{T_n,n}$ converges to quantities that behave
like their Gaussian experiment counterparts (\cref{cor:large_sample}).

Third, given suitably continuous $\eta_n$, we can show that the statistic 
\eqref{eq:real_wls_coef}, properly normalized, converges to a random variable $Z$ where $Z
\mid \Pi_{1:T}, T,
\eta \sim \Norm(0,1)$, where $\Pi_{1:T}, T, \eta$ are limits of their finite-sample
analogues (\cref{thm:standard_normal}). Since $Z$ is conditionally Gaussian, for
sufficiently continuous weighting functions $f$, we can show that $f$-weighted average
mis-coverage rates converge to zero uniformly over $\mathcal P$ (\cref{thm:appendix_main}).

Finally, we verify \cref{thm:maintext} by verifying that the continuity restrictions on
$\Pi_t, \one(T_n > t), \eta_n$ are satisfied given the assumptions in
\cref{sub:asymptotics}.

The uniformity aspects of the argument use the subsequencing argument in
\cite{andrews2011generic}. The key to doing so is the following lemma, which establishes
that certain subsequences exist given our assumptions on $\mathcal P$.

\begin{lemma}
\label{lemma:subsequence}
Under \cref{as:P_assumptions}, for any sequence $P_n \in \mathcal P$ and any subsequence
$n_r$ of $n$, there exists a further subsequence $n_p$ of $n_r$ such that, as $p \to
\infty$, \begin{enumerate}
    \item $\mu_{n_p} \equiv \mu(P_{n_p})  \to \mu \in [-\infty, \infty]^K$,
    \item $h_{n_p} \equiv \sqrt{n_p} (\mu_{n_p} - \max_k \mu_{n_p, k}) \to h \in
    [-\infty, 0]^K$,
    \item For all $k\neq k' \in [K]$, $\Delta\mu_{k,k',n_p} \equiv \sqrt{n}(\mu_{k,n_p} -
    \mu_
    {k',n_p}) \to \Delta \mu_{k,k'} \in [-\infty, \infty],$
    \item For all $k$, $\sigma^2_k(P_{n_p}) \to \sigma^2_k \in [C_1^{-1}, C_1]$,
    \item For all $k \in [K], t \in [T_0]$, $\hat\sigma^2_{tkn_p} \pto \sigma^2_k $.
\end{enumerate}
\end{lemma}

\begin{proof}

Fix the subsequence $n_r$ and the sequence $P_n$.  There exists a subsequence for which
(1)--(4) are satisfied since $[-\infty, \infty]$, $[-\infty, 0]$ (under the metric $d(x,y)
= \abs{\arctan(x) - \arctan(y)}$), and $[C_1^{-1}, C_1]^K$ are all compact.

Lastly, for (5), note that \eqref{eq:def_sigma2} can be written as \[
\sigma_{tkn}^2 = -(W_{tkn} - \mu_k(P_n))^2 + \frac{1}{\sum_{s=1}^t N_{skn}} \sum_{i
\in \mathcal I_{1:t}} D_{ik} (X_i^\obs - \mu_k(P_n))^2.
\]
We verify the following claims:
\begin{enumerate}[label=(\alph*)]
    \item We have  \[\P\bk{
 \abs[\bigg]{\sum_{i
\in \mathcal I_{1:t}} D_{ik} (X_i^\obs - \mu_k(P_n))} > \sqrt{n \log n}
    } \le \frac{C_1}{\log n}.
    \]
    \item  We have \[
\frac{1}{\sum_{s=1}^t N_{skn}} = O_p(1/n)
    \]
    \item For some $C > 0$, \[\P\bk{
 \abs[\bigg]{\sum_{i
\in \mathcal I_{1:t}} D_{ik} ((X_i^\obs - \mu_k(P_n))^2 - \sigma_k^2(P_n))} > \sqrt{n \log
n}
    } \le \frac{C}{\log n}.
    \]
\end{enumerate}
Given these claims, note that (a) and (b) imply that $W_{tkn} - \mu_k(P_n) = o_p
\pr{\sqrt{\frac{\log n}{n}}}$ and (b) and (c) imply that \[
\frac{1}{\sum_{s=1}^t N_{skn}} \sum_{i
\in \mathcal I_{1:t}} D_{ik} (X_i^\obs - \mu_k(P_n))^2 - \sigma^2_k(P_n) = o_p\pr{
\sqrt{\frac{\log n}{n}}}.
\]
As a result, $\sigma_{tkn}^2 - \sigma_{k}^2(P_n) = o_p(\sqrt{\log n / n})$. Hence
$\sigma_{tkn}^2 - \sigma_k^2 = o_p(1)$.

For (a), we note that the sum is a martingale. By Kolmogorov's inequality for martingales,
\[
\P\pr{\abs[\bigg]{\sum_{i
\in \mathcal I_{1:t}} D_{ik} (X_i^\obs - \mu_k(P_n))} > \sqrt{n \log n}
    }  \le \frac{1}{n \log n} \sum_{i=1}^n \sigma_k^2(P_n) \le \frac{C_1}{\log n}.
\]
For (b), we note that $\sum_s N_{skn} > N_{1kn} \sim \Bin(c_1 n, \pi_{1k})$. Since $N_
{1kn}$
grows at the rate of $n$, $
\frac{1}{\sum_{s=1}^t N_{skn}} = O_p(1/n).$ Similarly, we recognize that the sum in (c) is
again a martingale. To apply Kolmogorov's inequality, we need that the second moment of $
(X_i^\obs - \mu_k(P_n))^2 - \sigma^2_k(P_n)$ is finite. This is true due to the fourth
moment condition in \cref{as:P_assumptions}.
\end{proof}

\subsection{Behavior of random variables without a stopping time} 
We first analyze the
limiting behavior of the various random variables without a stopping
time. That is, we consider letting the experiment run until $T_0$.

Recall that $\Pi_{tn} = \diag(\Pi_{t1n},\ldots, \Pi_{tKn})$ collect the batch $t$
probabilities and that $\hat \Sigma_{tn}$ collect $\hat \sigma^2_{tkn}$. We assume that,
for some
function $\kappa_{t+1,n}$, \[
\Pi_{t+1,n} = \kappa_{t+1, n}\pr{X_{1:t,n}, \hat\Pi_{1:t,n}, \hat \Sigma_{tn}, n_{1:t}}.
\]
Note that \[
X_{tkn} = \begin{cases}
    \frac{1}{\sqrt{n_t} \hat\Pi_{tkn}} Y_{tkn} + \mu_{nk}, &\text{ if } \hat\Pi_{tkn} >
    0\\
    0 & \text{ otherwise}.
\end{cases}
\]
and thus is a function of $n_t, \hat\Pi_{tn}, Y_{tn}$. Therefore, we can write instead \[
\Pi_{t+1,n} = \kappa_{t+1, n}\pr{Y_{1:t,n}, \hat\Pi_{1:t,n}, \hat \Sigma_{tn}, n_{1:t}, \mu_n}.
\numberthis \label{eq:kappa_def}
\]

We require that $\kappa_{t+1}$ satisfies the following implication of the continuous
mapping theorem. This is more general than the conditions in \cref{sub:asymptotics}.

\begin{as}
\label{as:kappa_cmt}
The assignment algorithm $(\pi_1, \kappa_2,\ldots, \kappa_{T_0})$, for $\kappa_t$ in
\eqref{eq:kappa_def}, satisfies the following properties: Fix some $\epsilon > 0$,

\begin{enumerate}[wide]
    \item The assignment probabilities in the first batch are fixed and bounded away from
    zero: $\pi_{1k} > \epsilon > 0$ for every $k \in [K]$. Let $\Pi_1 = \pi_1$.

    \item Let $P_{n_m}$ denote a (sub)sequence of data-generating processes satisfying the
conclusion of \cref{lemma:subsequence}. Let $\Sigma = \diag(\sigma_1^2,\ldots,
\sigma_K^2)$.
Assume that, if \[Y_{1n_m} \dto Y_1 \sim \Norm(0, \Pi_1
\Sigma) \quad \hat\Pi_{1n_m} \pto \Pi_1,\numberthis \label{eq:precondition1}\]
then for some measurable function $\kappa_2$, we have the
following convergence jointly with \eqref{eq:precondition1} \[
\Pi_{2n_m} = \kappa_{2,n_m}\pr{Y_{1,n_m}, \hat\Pi_1, \hat \Sigma_{1n_m}, n_
{1}, \mu_{n_m}}
\dto
\kappa_2\pr{Y_1, \Pi_1; h, \sigma^2_{1:K}, c_1} \equiv \Pi_2.
\]

Moreover, assume that, if for some given $1 < t < T_0$, \[
\pr{Y_{1n_m}, \Pi_{2n_m}, \ldots, \Pi_{tn_m}, Y_{tn_m}} \dto (Y_1, \Pi_2, \ldots, \Pi_
{t}, Y_t) \quad \text{ and }\quad  \hat \Pi_{1:t, n_m} \pto \Pi_{1:t}, \numberthis
\label{eq:precondition2}
\]
where for all $s$, $Y_{s} \mid Y_{1:s-1}, \Pi_{2:s} \sim \Norm(0, \Pi_t \Sigma)$, then there
is a measurable function $\kappa_{t+1}$ such that we have the
following convergence jointly with \eqref{eq:precondition2} \[
\Pi_{t+1,n} = \kappa_{t+1, n}\pr{Y_{1:t,n}, \hat\Pi_{1:t,n}, \hat \Sigma_{tn}, n_{1:t}, \mu_n}
\dto
\kappa_{t+1}\pr{Y_{1:t}, \Pi_{1:t}; h, \sigma_{1:K}^2, c_{1:t}} \equiv
\Pi_{t+1}.
\numberthis \label{eq:condition2}
\]

\item For all $t = 2,\ldots,T_0$, the limiting random variable $\Pi_{t}$ is measurable
with
respect to \[
\pr{B_{t-1}(\Pi_{1:t-1}, c_{1:t-1}) \colvecb{3}{Y_1}{\vdots}{Y_{t-1}}, \Pi_{1:t-1}}
\]
for some conformable matrix $B_{t-1}(\Pi_{1:t-1}, c_{1:t-1})$ such that \[
B_{t-1}(\Pi_{1:t-1}, c_{1:t-1}) \colvecb{3}{\sqrt{c_1} \Pi_11_K}{\vdots}{\sqrt{c_{t-1}} \Pi_
{t-1}1_K} = 0
\]
almost surely.

\end{enumerate}
\end{as}

Let $G_{tn} = (\Pi_{1:t,n}, Y_{1:t,n})$ and let $G_t = (\Pi_{1:t}, Y_{1:t})$. The
following theorem characterizes the asymptotic behavior of $G_{T_0, n}$ and connects to
the normal model results. 

\begin{theorem}
\label{thm:main_convergence}
Under \cref{as:P_assumptions}, \cref{as:kappa_cmt}(1), and \cref{as:kappa_cmt}(2), for any
sequence $P_n
\in \mathcal P$ and any
subsequence $n_r$ of $n$, there exists a further subsequence $n_m$ such that, as
$m\to\infty$, \[
G_{T_0, n_m} \dto G_{T_0} = (\Pi_{1:T_0}, Y_{1:T_0})
\]
where, for all $t \in [T_0]$, \[
Y_{t} \mid Y_{1:t-1}, \Pi_{1:t} \sim \Norm(0, \Pi_t \Sigma).
\]
and $\Pi_{t+1}$ is measurable with respect to $\br{B_{s} \bY_{1:s} : s\le t}$.
\end{theorem}

\begin{proof}
Fix $P_n$ and a subsequence $n_r$ of $n$. By \cref{lemma:subsequence}, there exists a
further subsequence $n_m$ such that the conclusions of \cref{lemma:subsequence} hold.

We induct on $t$. Consider the base case of $G_{1n} = (\Pi_{1n}, Y_{1n}).$ By assumption
$\Pi_{1n}$ is fixed and constant in $n$. Thus it suffices to show that \[
Y_{1n_m} \dto \Norm(0, \Pi_1 \Sigma).
\]
This is true by the moment condition in \cref{as:P_assumptions}, where we apply the
central limit theorem in \cref{lemma:uclt}.

For the inductive case, suppose $G_{tn_m} \dto (\Pi_{1:t}, Y_{1:t})$ for some $1 \le t <
T_0$. We wish to show that $G_{t+1, n_m} \dto (\Pi_{1:t+1}, Y_{1:t+1})$. First, note that
by
a simple Chernoff bound in \cref{lemma:chernoff}, for all $s \le t$, \[
\hat\Pi_{s,n} - \Pi_{s,n} = o_p(1).
\]
Hence, $\hat\Pi_{s,n} \pto \Pi_s$ under the inductive hypothesis. Now,
by our assumptions
on the assignment algorithm in \cref{as:kappa_cmt}(2), \[
(\Pi_{1:t+1, n_m}, Y_{1:t, n_m}) \dto (\Pi_{1:t+1}, Y_{1:t}). \numberthis
\label{eq:pi_converges}
\]
By \cref{as:kappa_cmt}(2), $\Pi_{t+1}$ is measurable with respect to $\pr{B_{t} \bY_{1:t},
\Pi_{1:t}}$, which is further measurable with respect to $\br{B_{s} \bY_{1:s} : s \le t}$.

It remains to show that $Y_{t+1,n_m}$ converges jointly.
Let \[g(\cdot) \in \mathrm{BL}_1 \equiv \br{\sup_x |g(x) | \le 1, |g(x) - g(y)| \le d
(x,y)}\]
be a bounded Lipschitz function with respect to the distance metric $d$, where $d$ is the
product metric for $\Pi_{1:t+1}, Y_{1:t+1}$: \[
d(G_{t+1}, G_{t+1}') = \max_{s \in [t+1], k\in [K]} \br{|\Pi_{s,k} - \Pi'_{s,k}|, |Y_
{s,k} - Y_{s,k}'|}.
\]

Let $U_{n_m} = \Pi_{t+1, n_m}^{1/2}
Z_{n_m}$ for some $Z_{n_m} \sim \Norm(0,\Sigma(P_{n_m}))$
independently.
 Observe that \begin{align*}
&\abs[\big]{
    \E[g(G_{t,n_m}, \Pi_{t+1,n_m}, Y_{t+1,n_m})] - \E[g(G_{t,n_m}, \Pi_{t+1,n_m}, U_
    {n_m})] } \\ &= \abs[\big]{
    \E\bk{\E[g(G_{t,n_m}, \Pi_{t+1,n_m}, Y_{t+1,n_m}) \mid G_{t,n_m}, \Pi_{t+1,n_m}]} - \E
    \bk{\E[g(G_{t,n_m}, \Pi_{t+1,n_m}, U_{n_m}) \mid G_{t,n_m}, \Pi_{t+1,n_m}]}
}
\end{align*}
We claim that for some $c_m$ independent of $g$, \begin{align*}
&\abs[\big]{\E[g(G_{t,n_m}, \Pi_{t+1,n_m}, Y_{t+1,n_m}) \mid G_{t,n_m}, \Pi_{t+1,n_m}] - \E
[g(G_
{t,n_m}, \Pi_{t+1,n_m}, W_{n_m}) \mid G_{t,n_m}, \Pi_{t+1,n_m}]} \\
&\le c_m \to 0
\numberthis
\label{eq:conditional_converge}
\end{align*}
with probability one. Assuming this claim, we would show that \[
\sup_{g \in \mathrm{BL}_1}\abs[\big]{
    \E[g(G_{t,n_m}, \Pi_{t+1,n_m}, Y_{t+1,n_m})] - \E[g(G_{t,n_m}, \Pi_{t+1,n_m}, W_{n_m})]
} \to 0.
\]
Observe that by construction of $U_{n_m}$ and \eqref{eq:pi_converges}, for some $Y_{t+1}$,
\[
(G_{t,n_m}, \Pi_{t+1, n_m}, U_{n_m}) \dto \pr{G_t, \Pi_{t+1}, Y_{t+1}}.
\]
Equivalently, \[
\sup_{g \in \mathrm{BL}_1} \abs[\big]{
    \E[g(G_{t,n_m}, \Pi_{t+1,n_m}, U_{n_m})] - \E[g(G_{t}, \Pi_{t+1}, Y_{t+1})]
} \to 0.
\]
Therefore, by the triangle inequality, \[
\sup_{g\in \mathrm{BL}_1} \abs[\big]{
    \E[g(G_{t,n_m}, \Pi_{t+1,n_m}, Y_{t+1,n_m})] - \E[g(G_{t}, \Pi_{t+1}, Y_{t+1})]
} \to 0.
\]
Equivalently, \[
G_{t+1,n_m} \dto G_{t+1}.
\]

Thus, to conclude the proof, we show \eqref{eq:conditional_converge}. For given values of
$G_
{t,n_m}, \Pi_{t+1, n_m} = (G, \Pi)$, let \[
h(y) = g(G,\Pi, y).
\]
Observe that $h(\cdot)$ is bounded and Lipschitz with respect to $d_y(y_1, y_2) = \max_k
|y_{1k} - y_{2k}|$. Thus, \eqref{eq:conditional_converge} amounts to the following: \[
\abs[\big]{\E_{Y_{t+1,n_m} \mid \Pi} [h(Y_{t+1, n_m})] - \E_{U_{n_m} \mid \Pi}[h(U_
{n_m})]} \le
c_m. \]
Note that by \cref{lemma:uclt} and \cref{as:P_assumptions}, for $\mathsf{prok}$ the
Prokhorov metric (see \cref{lemma:uclt} and Proposition A.5.2 in \citet{van1996weak}),\[
\abs[\big]{\E_{Y_{t+1,n_m} \mid \Pi} [h(Y_{t+1, n_m})] - \E_{U_{n_m} \mid \Pi}[h(U_
{n_m})]} \lesssim \sup_{\Pi \in \Delta^{K-1}, P\in \mathcal P} \mathsf{prok}\pr{Y_{n_m}
(\Pi, P), U_{n_m}
(\Pi, P)} \le c_m \to 0,
\]
where \[
Y_{k,n_m}(\Pi, P) = \frac{1}{\sqrt{n_t}} \sum_{i \in \mathcal I_t} D_{ik} X_i(k) \quad D_
{i} \sim \Mult(1, \Pi), X_i(\cdot) \sim P
\]
and \[U_{n_m}
(\Pi, P) \sim \Norm(0, \Pi \Sigma(P)).
\]
This concludes the proof.
\end{proof}

\subsubsection{Stopping time}

So far, we have shown that $G_{T_0, n_m}$ converges. Now, we consider the following
assumption on the stopping time $T_n$. Let $\Xi_{t, n} = \one(T_n > t)$. We assume that
$\Xi_{t,n}$ satisfies conditions analogous to those satisfied by $\kappa_t$ in
\cref{as:kappa_cmt}. 

\begin{as}
\label{as:T_assumptions}
We assume that for all $t = 1,\ldots, T_0$, \begin{enumerate}[wide]
    \item $\Xi_{t,n} = \Xi_{t,n}\pr{X_{1:t-1,n}, \hat\Pi_{1:t-1,n}, \hat\Sigma_{t-1, n},
    n_{1:t-1}}
    =
    \Xi_{t,n}\pr{
    Y_{1:t-1,n}, \hat\Pi_{1:t-1,n}, \hat\Sigma_{t-1,n}, n_{1:t-1}, \mu_n
    }$ is measurable with respect to $(Y_{1:t-1,n}, \hat\Pi_{1:t-1,n})$ and $\Xi_{T_0,n}
    = 0$
    a.s.
    \item Let $P_{n_m}$ denote a subsequence of the data-generating process satisfying
    the conclusion of \cref{lemma:subsequence}. Then, if \[
(\hat\Pi_{1:T_0,n}, Y_{1:T_0,n}) \dto G_{T_0} = (\Pi_{1:T_0}, Y_{1:T_0})
    \]
    then for all $t=1,\ldots, T_0$, we have the following convergences:\[
\Xi_{t,n_m}\pr{
    Y_{1:t-1,n}, \hat\Pi_{1:t-1,n}, \hat\Sigma_{t-1,n}, n_{1:t-1}, \mu_n
    } \dto \Xi_{t}\pr{Y_{1:t-1}, \Pi_{1:t-1}; h, \sigma_{1:K}^2, c_{1:t-1}}
    \]
    and \[
T_n \dto T \equiv \argmax_{t=1,\ldots, T_0} \, (\Xi_{t-1} - \Xi_{t})
    \]
    jointly.
    \item For all $t = 1,\ldots, T_0-1$, the limiting random variable $\Xi_{t}\pr{Y_
    {1:t-1},
    \Pi_{1:t-1}; h, \sigma_{1:K}^2, c_{1:t-1}}$ is measurable with respect to the same
    random variables as $\Pi_{t}$ is in \cref{as:kappa_cmt}(3). That is, $\Xi_t$ is
    measurable with respect to $\br{B_s \bY_{1:s} : s\le t-1}$ for matrices $B_s$ in
    \cref{as:kappa_cmt}(3).

\end{enumerate}
\end{as}

Note that we can write \[L_n - \lambda_n'\mu = \sum_{s=1}^{T_0} \underbrace{(\Xi_{s-1,n} -
\Xi_{s,n})}_{\one(T_n=s)}\sum_{t=1}^{s-1} \sum_
{k=1}^{K} \frac{n_t \hat\Pi_{tkn}}{\sqrt{n} \hat\sigma_{skn}^2} (X_{tkn} - \mu)\]
and \[
Y_{T_n, n} = \sum_{s=1}^{T_0} (\Xi_{s-1,n} - \Xi_{s,n}) Y_{sn}.
\]
The following corollary shows the asymptotic analogue of \cref{thm:leftover}.
\begin{cor}
\label{cor:large_sample}
Additionally under \cref{as:kappa_cmt}(3) and \cref{as:T_assumptions}, on the subsequence $n_m$,
\[\colvecb{2}{L_{n_m} - \lambda_n'\mu_{n_m}}{Y_{T,n_m}} \dto \colvecb{2}{
1'\Sigma^{-1} Y_{:T-1}}{Y_{T}},\]
where, in the limit, \[
\colvecb{2}{1'\Sigma^{-1} Y_{:T-1}}{Y_{T}} \mid \pr{T, \br{B_{s} \bY_{1:s} : s \le T-1}}
\sim
\Norm\pr{
    \colvecb{2}{0}{0}, \begin{bmatrix}
       1'\Sigma^{-1} \Pi_{:T-1}  & 0 \\
       0 & \Sigma\Pi_T
    \end{bmatrix}
}
\]
\end{cor}

\begin{proof}
The weak convergence is by the continuous mapping theorem and \cref{lemma:convergence_subscript}, noting that
$\hat\sigma_
{skn}
\to \sigma_k^2$ and $n_t/n \to c_t$ for all $s, t$. 

The second claim is a claim on the properties of the limiting distribution. Note that
\cref{as:kappa_cmt}(3) and \cref{as:T_assumptions}(3) ensure that the limiting assignment
probabilities and the limiting stopping time are functions of certain differences in
$Y_{1:t}$. As a result, we can apply \cref{lemma:leftover} directly to finish the proof.
\end{proof}

\subsection{Behavior of the test statistic } We now rewrite the empirical analogue of the
least-squares coefficient
\eqref{eq:real_wls_coef}: \[
S_n^\star = \pr{\frac{\lambda_n \lambda'_n / \sqrt{n}}{\lambda'_n1} + \frac{n_{T_n}}{n}
\hat\Sigma^{-1}_n \hat \Pi_{T_nn} }^+
\pr{\frac{\lambda_n}{\lambda_n'1} \frac{L_n}{\sqrt{n}} + \frac{n_{T_n}}{n} \hat\Sigma^
{-1}_n
\hat\Pi_
{T_nn} X_{T_nn}}.
\]
Let $\eta_n = \eta(\hat\Pi_{1:T_n,n})$ be a direction of inference. We additionally assume that when the experiment stops, the probability of the last batch
is pruned so as to exclude $(0, \epsilon)$.

\begin{as}[Pruning]
\label{as:pruning_general}
For some $\epsilon > 0$, the random variable \[\Pi_{T_n,n} = \kappa_{T_n, n}
    (\cdot) = \sum_{s=1}^T (\Xi_{s-1,n} -
\Xi_{s,n}) \kappa_{s,n}(\cdot)\] takes values outside
    of $\br{q \in \Delta^{K-1} : q_k \in (0, \epsilon) \text{ for some $k \in  [K]$}}$
    with probability one.
\end{as}

We assume that $\eta(\cdot)$ is suitably continuous. 
\begin{as}
\label{as:eta_assumptions}

Under a subsequence $n_m$ of $n$ that satisfies the consequences of 
\cref{lemma:subsequence,thm:main_convergence}, if \[
(\hat\Pi_{1:T_0, n}, Y_{1:T_0, n}) \dto G_{T_0} \text{ and } T_n \dto T 
\]
Then 
$\eta_n = \eta\pr{
    Y_{1:T_n,n}, \hat\Pi_{1:T_n,n}, T_n, \hat\Sigma_{T_n,n}, n_{1:T_n}, \mu_n
    } \dto \eta(Y_{1:T}, \Pi_{1:T}, T, \Sigma, c_{1:T})$ along the subsequence $n_m$.
    Moreover, $\eta$ is measurable with respect to $(T, \br{B_{s} \bY_{1:s} : s \le T-1})$
    for matrices $B$ in \cref{as:kappa_continuity}. 
We also assume that
$\norm{\eta_{n_m}} > \epsilon > 0$ with probability one.
Moreover, $\eta(T, \Pi_{1:T})$ is in the column space of \[
\frac{\lambda \lambda'}{1' \lambda} + c_T \Sigma^{-1}\Pi_T
\]
almost surely and $\eta_{n_m}$ is in the column space of \[
\frac{\lambda_{n_m} \lambda'_{n_m} / \sqrt{n_m}}{\lambda'_{n_m}1} + \frac{n_T}{n_m}
\hat\Sigma^{-1}_{n_m} \hat \Pi_{Tn_m}
\]
almost surely for all sufficiently large $m$.
\end{as}

Under this assumption, note that \[
\sqrt{n}(\eta_n'S_n^\star - \eta_n'\mu_n) = \eta_n'\pr{\frac{\lambda_n \lambda'_n /
\sqrt{n}}{\lambda'_n1} + \frac{n_T}{n}
\hat\Sigma^{-1}_n \hat \Pi_{Tn} }^+
\pr{\frac{\lambda_n}{\lambda_n'1} \pr{L_n - \lambda_n' \mu_n} + \sqrt{c_T}
\hat\Sigma^{-1}_n Y_{Tn}}
\]
since $\eta'_n Q^+Q \mu = \eta_n'\mu$ if $\eta_n$ is in the column space of the
symmetric matrix $Q$. Let \[\hat\sigma^2_{\tau,n} = \eta_n' \pr{\frac{\lambda_n \lambda'_n /
\sqrt{n}}{\lambda'_n1} + \frac{n_T}{n}
\hat\Sigma^{-1}_n \hat \Pi_{Tn} }^+ \eta_n.\]
Then, we obtain the uniform convergence of the test statistic.
\begin{theorem}
\label{thm:standard_normal}
Under \cref{as:P_assumptions,as:kappa_cmt,as:T_assumptions,as:pruning_general,as:eta_assumptions}, for every sequence $P_n
\in \mathcal P$ and every subsequence $n_r$ of $n$, there exists a further subsequence
$n_m$ such that, jointly, \[
\frac{\sqrt{n_m}(\eta_{n_m}'S_{n_m}^\star - \eta_{n_m}'\mu_{n_m})}{\hat \sigma_{\tau,
n_m}} \dto Z \quad (\eta_{n_m}, \Pi_{1:T_{n_m}}, T_{n_m}) \dto (\eta, \Pi_{1:T}, T)
\]
where $
Z \mid (\eta, \Pi_{1:T}, T) \sim \Norm(0,1). 
$
\end{theorem}

\begin{proof}
Let $n_m$ be the subsequence that satisfies \cref{cor:large_sample}. Under
\cref{as:eta_assumptions}, $\eta_{n_m} \dto \eta$ jointly with $(\Pi_{1:T_{n_m}}, T_{n_m})
\dto (\Pi_{1:T}, T)$.
By the continuous mapping theorem and \cref{lemma:convergence_subscript}, \[
\frac{\lambda_{n_m} \lambda'_{n_m} / \sqrt{n_m}}{\lambda'_{n_m}1} + \frac{n_{T_{n_m}}}
{n_m}
\hat \Sigma^{-1}_{n_m} \hat \Pi_{T_{n_m}n_m}  \dto \frac{\lambda' \lambda}{\lambda'1} +
c_T
\Sigma^
{-1}
\Pi_T
\]
jointly with $G_{T_{n_m},n_m} \dto G_{T, n_m}$.

Next, \cref{lemma:moore_penrhose_converges} shows that the convergence in the above
display is preserved by the Moore-Penrose pseudoinverse, despite the discontinuity in the
pseudoinverse. This is because of pruning.

Next, we show that the standard error is positive in the limit. Note that there exists
some $c > 0$ such that, uniformly over $\mathcal P$, \[
\eta'\pr{\frac{\lambda' \lambda}{\lambda'1} + c_T
    \Sigma^
  {-1}
  \Pi_T}^+\eta \ge \lambda_{\max}^{-1}\pr{\frac{\lambda' \lambda}{\lambda'1} + c_T
    \Sigma^
  {-1}
  \Pi_T} \norm{\eta} > c > 0
\]
where $\lambda_{\max}(A)$ is the maximum eigenvalue of $A$.
Thus, continuous mapping theorem implies that \begin{align*}
\frac{\sqrt{n}(\eta_n'S_n^\star - \eta_n'\mu_n)}{\hat\sigma_{\tau,n}} \dto \frac{\eta'\pr{\frac{\lambda' \lambda}{\lambda'1} + c_T
    \Sigma^
  {-1}
  \Pi_T}^+\pr{\frac{\lambda}{\lambda'1} 1'\Sigma^{-1} Y_{:T-1} + \sqrt{c_T} \Sigma^{-1}
  Y_T}}
{\sqrt{\eta'\pr{\frac{\lambda' \lambda}{\lambda'1} + c_T
    \Sigma^
  {-1}
  \Pi_T}^+\eta}} \equiv Z.
\end{align*}
Note that by \cref{cor:large_sample} \begin{align*}
\eta'\pr{\frac{\lambda' \lambda}{\lambda'1} + c_T
    \Sigma^
  {-1}
  \Pi_T}^+\pr{\frac{\lambda}{\lambda'1} 1'\Sigma^{-1} Y_{:T-1} + \sqrt{c_T} \Sigma^{-1}
  Y_T} \mid \eta, \Pi_{1:T}, T \\\sim \Norm\pr{0, \eta'\pr{\frac{\lambda' \lambda}
  {\lambda'1} + c_T
    \Sigma^
  {-1}
  \Pi_T}^+\eta}
\end{align*}
 Thus, \[
\frac{\eta'\pr{\frac{\lambda' \lambda}{\lambda'1} + c_T
    \Sigma^
  {-1}
  \Pi_T}^+\pr{\frac{\lambda}{\lambda'1} 1'\Sigma^{-1} Y_{:T-1} + \sqrt{c_T} \Sigma^{-1}
  Y_T}}
{\sqrt{\eta'\pr{\frac{\lambda' \lambda}{\lambda'1} + c_T
    \Sigma^
  {-1}
  \Pi_T}^+\eta}} \mid \eta, \Pi_{1:T}, T \sim \Norm(0,1).
 \]
 This completes the proof since $\eta, \Pi_{1:T}, T$ is measurable with respect to \[
(T, \br{B_s \bY_{1:s} : s \le T-1}). 
 \]
\end{proof}
\subsection{Conditional coverage performance}
\begin{as}
\label{as:weighting_assumptions}
Let $f(\Pi_{1:T}, T,\eta) < C < \infty$ be a bounded function such that along all
subsequences $P_{n_m} \in \mathcal P$, if \[ (\hat\Pi_{1:T_{n_m},n_m}, T_{n_m},
\eta_{n_m}) \dto \pr{\Pi_{1:T}, T, \eta}
\]
then \[
f(\hat\Pi_{1:T_{n_m},n_m}, T_{n_m}, \eta_{n_m}) \dto f\pr{\Pi_{1:T}, T, \eta}.
\]
\end{as}

\begin{theorem}
\label{thm:appendix_main}
Suppose
\cref{as:P_assumptions,as:kappa_cmt,as:T_assumptions,as:pruning_general,as:eta_assumptions} hold and let $\alpha \in (0,0.5)$. Let $f(\Pi_{1:T}, T, \eta) > 0$ be a function satisfying \cref{as:weighting_assumptions}. Then, 
\[
\limsup_{n\to\infty} \sup_{P \in \mathcal P} \abs[\bigg]{\E_P \bk{\one \pr{ \abs{\eta_n'S_
{n}^\star -
\eta'_n\mu(P)}  >
z_{1-\alpha/2} \frac{\hat\sigma_{\tau, n}}{\sqrt{n}}} f(\Pi_{1:T_n,n}, T_n, \eta_{n})} -
\alpha \E_P[f(\Pi_{1:T_n, n}, T_n, \eta_{n})]} =
0.
\]
where $z_{1-\alpha/2} = \Phi^{-1}(1-\alpha/2)$.
\end{theorem}
\begin{proof}
Suppose to the contrary, for some $\epsilon > 0$, \[
\limsup_{n\to\infty} \sup_{P \in \mathcal P} \abs[\bigg]{\E_P \bk{\one \pr{ \abs{\eta_n'S_
{n}^\star -
\eta'_n\mu(P)}  >
z_{1-\alpha/2} \frac{\hat\sigma_{\tau, n}}{\sqrt{n}}} f(\Pi_{1:T_n,n}, T_n, \eta_{n})} -
\alpha \E_P[f(\Pi_{1:T_n, n}, T_n, \eta_{n})]} > 2\epsilon.\]
Then there exists a subsequence $P_{n_r}$ where for all $r > 0$, \[
\limsup_{n\to\infty} \sup_{P \in \mathcal P} \abs[\bigg]{\E_P \bk{\one \pr{ \abs{\eta_n'S_
{n}^\star -
\eta'_n\mu(P)}  >
z_{1-\alpha/2} \frac{\hat\sigma_{\tau, n}}{\sqrt{n}}} f(\Pi_{1:T_n,n}, T_n, \eta_{n})} -
\alpha \E_P[f(\Pi_{1:T_n, n}, T_n, \eta_{n})]} > \epsilon. \numberthis \label{eq:n_r_contradiction}
\]

By \cref{thm:standard_normal}, we can extract a further subsequence $n_m$ along which \[
\sqrt{n_m}\frac{\abs{\eta_{n_m}'S_{n_m}^\star - \eta'_{n_m}\mu(P_{n_m})}}{
\hat\sigma_{\tau, n_m} } \dto Z \numberthis \label{eq:n_m_normal}
\]
and \[
(\eta_{n_m}, \Pi_{1:T_{n_m}}, T_{n_m}) \dto (\eta, \Pi_{1:T}, T)
\]
where \[
Z \mid (\eta, \Pi_{1:T}, T) \sim \Norm(0,1). 
\]

By \cref{as:weighting_assumptions}, we additionally have \[
f(\Pi_{1:T_{n_m}, n_m}, T_{n_m}, \eta_{n_m}) \dto f(\Pi_{1:T}, T, \eta). 
\]
Thus \[
\one\pr{\sqrt{n_m}\frac{\abs{\eta_{n_m}'S_{n_m}^\star - \eta'_{n_m}\mu(P_{n_m})}}{
\hat\sigma_{\tau, n_m} } < z_{1-\alpha/2}} f(\Pi_{1:T_{n_m}, n_m}, T_{n_m}, \eta_{n_m}) \dto \one(Z < z_{1-\alpha/2}) f(\Pi_{1:T}, T, \eta). 
\]
Since $f$ is bounded, we have that the expectations also converge. Note that \[\E[\one(Z <
z_{1-\alpha/2}) f(\Pi_{1:T}, T, \eta)] = \alpha \E[f(\Pi_{1:T}, T, \eta)].\] Therefore,
along this subsequence $n_m$, the limit is zero, contradicting
\eqref{eq:n_r_contradiction}.
\end{proof}

\subsection{Auxiliary lemmas}

\begin{lemma}[Uniform central limit theorem via Lyapunov condition]
\label{lemma:uclt}
    Let $\pi \in \Delta^{K-1}$ be a vector of probabilities. Let $X_i,\ldots, X_n \iid
    Q$ be $\R^K$-valued random vectors for some distribution $Q \in \mathcal Q$ in a
    family $\mathcal Q$, with $\E_Q[X] = \mu
    (Q)$. Let $\Sigma(Q)= \diag(\var_Q(X_1), \ldots, \var_Q(X_K))$.
    Let $Y(Q, \pi) \in
    \R^K$ be such that its $k$\th{} coordinate is \[
Y_{k} (Q, \pi) = \frac{1}{\sqrt{n}}\sum_{i=1}^n D_{ik}(X_{ik} - \mu_k(Q)) \quad D_i =
[D_{i1},\ldots, D_
{iK}]' \iid \Mult(1, \pi).
    \]
    Correspondingly, write $Y(Q, \pi) = \frac{1}{\sqrt{n}} \sum_{i=1}^{n} Y_i$.
    Suppose that there exists $c > 0$ such
    that \[\sup_{Q \in \mathcal Q} \E_Q \bk{
        \norm{X}^{2+c}} < C < \infty,\] where $\norm{\cdot}$ denotes the Euclidean norm.
    Then, there exists $c_n \to 0$ dominating the \emph{Prokhorov metric}:   \[
    \sup_{Q \in \mathcal Q} \sup_{\pi \in \Delta^{K-1}} \mathsf{prok}(Y(Q, \pi), \Norm
    (0, \Sigma(Q)\Pi))
    \le c_n
    \to 0,
    \]
    where $\mathsf{prok}$ denotes the Prokhorov metric and $\Pi = \diag(\pi)$. Since the
    Prokhorov metric
    dominates the bounded Lipschitz metric, the same statement hold for the latter.
\end{lemma}
\begin{proof}
    This is a straightforward application of Proposition A.5.2 in
    \cite{van1996weak}.
    Observe that, since the entries of $D_i$ are bounded by $1$, $Y_i$ satisfies the
    Lyapunov condition \[
    \norm{Y_i} \le \norm{X_i-\mu(Q)} \text{ almost surely} \implies \E_{Q, \pi}\norm{Y_i}^
    {2+c}
    < C
    \]
    and $
    \var_{Q, \pi}(Y_i) = \Sigma(Q)\diag(\pi).
    $

    For a given $Q$, let \begin{align*}
        g(t, Q, \pi) &= \E_{Q,\pi}\norm{Y_i}^2 \one(\norm{Y_i} \ge t)
     \\ &\le \E[\norm{Y_i}^{2+c} / t^c] \minwith \E[\norm{Y_i}^2]
 \le C (t^{-c} \minwith 1)
    \end{align*}

    Thus $g(\delta\sqrt n, Q, \pi) \lesssim \frac{1}{\delta^{c} n^{c/2}} \minwith 1$.
    Then Proposition A.5.2 in \cite{van1996weak} implies that up to constants that
    only depend on the moment bound $C$ and $K$, 
    \begin{align*}
    \mathsf{prok}\pr{Y(Q, \pi), \Norm(0, \Sigma(Q, \pi))} &\lesssim \frac{1}{\delta^{2+c}
    n^{c/2}} \maxwith \delta  + \frac{1}{\delta^{c/3} n^{c/6}}  + \delta^{1/4} \pr{1 +
    \sqrt{\abs{\log \delta}}}
    \end{align*}
    The expression on the right hand side vanishes if we choose $\delta = n^{-\frac{c}{4
    (2+c)}}$. THis completes the proof. 
\end{proof}

\begin{lemma}
\label{lemma:chernoff}
    Let $N \sim \Bin(n, p)$, then with probability at least $1-\Delta_n$,
    \[
    \abs[\bigg]{\frac{N}{n} - p} \le d_n
    \]
    where $d_n, \Delta_n \to 0$ are sequences independent of $p$. Hence, by a union
    bound, for any $K$ variables $N_1,\ldots, N_K$ where $N_k \sim \Bin(n, p_k)$, \[
    \P\bk{
        \forall k : \abs[\bigg]{\frac{N_k}{n} - p_k} \le d_n
    } \ge 1- K \Delta_n.
    \]
\end{lemma}
\begin{proof}
    Exercise 2.3.5 in \cite{vershynin2010introduction} shows the following
    Chernoff bound: For some absolute $c > 0$, \[
    \P\bk{
        \abs[\bigg]{\frac{N}{n} - p} \ge \delta p
    } \le 2 e^{-c np \delta^2} %
  \]
    If we choose $\delta = d_n / p$, we have \[
     \P\bk{
        \abs[\bigg]{\frac{N}{n} - p} \ge d_n
    } \le 2 e^{-c nd_n^2 / p} \le 2e^{-cnd_n^2}
    \]
    Pick, for instance, $d_n = \log(n)/\sqrt{n} + 1/n$ and $\Delta_n = 1-2n^{-c}$ to
    complete the proof.
\end{proof}

\begin{lemma}
\label{lemma:convergence_subscript}
Under our setup, consider a subsequence $n_m$.
Let $V_{tn_m}$ be certain random variables indexed by $t=1,\ldots, T_0$ and $n_m$. Suppose
\[
(V_{1n_m},\ldots, V_{T_0,n_m}, \Xi_{1n_m},\ldots, \Xi_{T_0, n_m}, T_{n_m}) \dto 
(V_1,\ldots, V_{T_0}, \Xi_1,\ldots, \Xi_{T_0}, T)
\]
Then $
V_{T_{n_m}, n_m} \dto V_T.
$
\end{lemma}

\begin{proof}
Note that, by continuous mapping theorem, \[
V_{T_{n_m}, n_m} = \sum_{t=1}^{T_0} (\Xi_{t-1, n_m} - \Xi_{t, n_m}) V_{t, n_m} \dto \sum_
{t=1}^{T_0} (\Xi_{t-1} - \Xi_{t}) V_{t} = V_T.
\]
This concludes the proof. 
\end{proof}

\begin{lemma}
\label{lemma:penrose}
Let $M_n, M$ be random positive semi-definite matrices such that $M_n \dto
M$. Suppose that there exists sets $\mathcal M_n, \mathcal M$ such that (a) for all $n$,
$\P(M_n \in
\mathcal
M_n) = \P(M \in \mathcal M) = 1$ and that (b) for every deterministic sequence $v_n \to v$
where $v_n \in \mathcal M_n$ and $v \in \mathcal M$, we have that for all sufficiently
large $n$ $\rank(v_n) = \rank(v)$. Then $M_n^+ \dto M^+$ where $(\cdot)^+$ is the
Moore-Penrose pseudo-inverse.
\end{lemma} 

\begin{proof}
Since $\rank(v_n) = \rank(v)$ for every sequence eventually, we have that $v_n^+ \to v^+$
for all sequences \citep{rakovcevic1997continuity}. Hence by the extended continuous
mapping theorem, $M_n^+ \dto M^+.$
\end{proof}

\begin{lemma}
\label{lemma:moore_penrhose_converges}
Under the conditions of \cref{thm:standard_normal} and in its proof, \[
\pr{\frac{\lambda_{n_m} \lambda'_{n_m} / \sqrt{n_m}}{\lambda'_{n_m}1} + \frac{n_{T_{n_m}}}
  {n_m}
  \hat \Sigma^{-1}_{n_m} \hat \Pi_{T_{n_m}n_m}}^+  \dto \pr{\frac{\lambda' \lambda}
  {\lambda'1} +
  c_T
  \Sigma^
  {-1}
  \Pi_T}^+
\]
jointly. 
\end{lemma}
\begin{proof}
We verify this statement with \cref{lemma:penrose}. Let \[
M_{n_m} = \frac{\lambda_{n_m} \lambda'_{n_m} / \sqrt{n_m}}{\lambda'_{n_m}1} + \frac{n_{T_
{n_m}}}
  {n_m}
  \hat \Sigma^{-1}_{n_m} \hat \Pi_{T_{n_m}n_m}
\]
and $M = \frac{\lambda' \lambda}
  {\lambda'1} +
  c_T
  \Sigma^
  {-1}
  \Pi_T$ where $M_{n} \dto M$ (along a subsequence) is already shown. Consider a fixed
  sequence $v_{m} \to v$. Note
  that $v_m$ is of the form \[
\frac{\ell_m \ell_m' / \sqrt{n_m}}{\ell_m' 1} + \frac{n_{T_m}}{n_m} \Omega^{-1}_{m} \pi_
{T_{m},m}
  \]
Similarly, we can write \[
v = \frac{\ell \ell' }{\ell' 1} + c_{T^*} \Omega^{-1} \pi_{T^*}
  \]
  Note that $v_m \to v$ implies that $\frac{\ell_m \ell_m' / \sqrt{n_m}}{\ell_m' 1} \to 
  \frac{\ell \ell' }{\ell' 1}$, by inspecting the convergence of the off-diagonal
  elements. Hence \[
\frac{n_{T_m}}{n_m} \Omega^{-1}_{m} \pi_
{T_{m},m} \to c_{T^*} \Omega^{-1} \pi_{T^*}
  \]

  Since $\Pi_{T_{n_m}, n_m} \dto \Pi_T$, by \cref{as:pruning_general}, we can restrict to
  considering $\pi_{T^*}$ whose entries are never within $(0, \epsilon)$. Thus the
  diagonal entries of $c_{T^*} \Omega^{-1} \pi_{T^*}$ are never within $(0, \min_{t \in 
  [T_0]} c_t \frac{1}{C_1} \epsilon)$ by \cref{as:P_assumptions}. Hence, if \[
\frac{\ell_m \ell_m' / \sqrt{n_m}}{\ell_m' 1} \to 
  \frac{\ell \ell' }{\ell' 1} \quad \frac{n_{T_m}}{n_m} \Omega^{-1}_{m} \pi_
{T_{m},m} \to c_{T^*} \Omega^{-1} \pi_{T^*},
  \]
  then for all sufficiently large $m$, we must have that (a) the number and location of
  zeros in the diagonal of $\frac{n_{T_m}}{n_m} \Omega^{-1}_{m} \pi_
{T_{m},m}$ remain fixed and (b) $\ell_m$ has
strictly positive entries. 

It suffices to verify that for a strictly positive vector $a \in \R^K$ and diagonal matrix
$A \succeq 0$ with rank $r$, $aa' + A$ has rank equal to $K \minwith (r+1)$. Given this
claim, since the number of zeros in the diagonal of $\frac{n_{T_m}}
{n_m} \Omega^{-1}_{m} \pi_
{T_{m},m}$ remain unchanged for sufficiently large $m$, its rank is constant. Thus the
rank of $v_m$ is also constant for sufficiently large $m$. 

To verify this claim, we will show that $a$ and $\br{e_k: A_{kk} \neq 0}$ span the
column space of $aa' + A$, where $e_k$ is the $k$\th{} standard basis vector. First, note
that if $A \succ 0$ is positive definite, then $aa' + A$ is full-rank, and the claim is
true. Otherwise, there is some $k$ for which $A_{kk} = 0$. Then $(aa' + A) e_k = (a'e_k)
a$. Since $a$ has strictly positive entries, this shows that $a$ is in the column space of
$aa'+A$. Let $\ell$ be such that $A_{\ell\ell}\neq 0$. Then $(aa' + A) e_\ell = A_
{\ell\ell} e_\ell + (a'e_\ell) a$ is a linear combination of $e_\ell$ and $a$. Since $a$
is in the column space of $aa'+A$, then $e_\ell$ is also in the column space. Lastly,
since every vector in the column space is in the span of $a$ and $\br{e_k: A_{kk} \neq
0}$: $(aa'+A)u = (a'u)a + \sum_{k: A_{kk}\neq0} A_{kk} u_k e_k$, we conclude that $a$ and
$\br{e_k: A_{kk} \neq 0}$ span the column space of $aa'+A$. We finish the proof by noting
that the dimension of $a$ and $\br{e_k: A_{kk} \neq 0}$ is $K \minwith (r+1)$.
\end{proof}

\subsection{Verification of \cref{thm:maintext}}
\label{sub:asymptotics-proof}

\thmmaintext*

\begin{proof}

The theorem follows from \cref{thm:appendix_main} pending verification of
 \cref{as:eta_assumptions,as:kappa_cmt,as:T_assumptions,as:weighting_assumptions,as:pruning_general}.

Note that our setup means that \cref{as:kappa_cmt}(1) and \cref{as:T_assumptions}(1) are
satisfied. \Cref{lemma:kappa} verifies \cref{as:kappa_cmt}(2--3) and 
\cref{as:T_assumptions}(2--3). \Cref{as:kappa_pruning} implies \cref{as:pruning_general}
directly. \Cref{lemma:eta} verifies \cref{as:eta_assumptions}. 
Finally, since $\eta_n$ and $T_n$ both have finite support (with discrete topology), for
any $\eta_1, T_1$ in the support, $f
(\eta, T) = \one(\eta
= \eta_1, T = T_1)$ satisfies \cref{as:weighting_assumptions}.
\end{proof}

\begin{lemma}
\label{lemma:kappa}

If the assignment algorithm satisfies \cref{as:kappa_continuity,as:Xi_continuity} then it
satisfies \cref{as:kappa_cmt}(2--3) and \cref{as:T_assumptions}(2--3). 

\end{lemma}

\begin{proof}

Note that we can rewrite \[
\sqrt{n} W_{t,n} = \underbrace{\pr{\sum_{s=1}^t \frac{n_t}{n} \hat\Pi_{sn}}^{-1}}_
{\hat\Pi_
{:t,n}^{-1}} \underbrace{\sum_
{s=1}^t \sqrt{
\frac{n_t}{n}} Y_{sn}}_{Y_{:t,n}} + \underbrace{\sqrt{n} \mu_{n}}_{h_n + \max_\ell \mu_
{n,\ell}}
\]
Hence we can equivalently write \[
\kappa_{t+1}\pr{\sqrt{n} W_{t,n}, \hat\Sigma_{tn} \hat\Pi_{:t,n}^{-1} } = \kappa_{t+1}
\pr{
    \hat\Pi_{:t,n}^{-1} Y_{:t,n} + h_n, \hat\Sigma_{tn} \hat\Pi_{:t,n}^{-1}
}
\]
and similarly for $\Xi_{t+1}$. 

With respect to the subsequence $P_{n_m}$, let $J = \br{j : h_{n_m, j} > -\infty }$. We
restrict to the subsequence $n_m$ and index by $n$ instead.

To verify \cref{as:kappa_cmt}(2--3) and \cref{as:T_assumptions}(2--3), we induct on $t$.
We will verify that
\begin{itemize}
    \item The statements in \cref{as:kappa_cmt}(2--3) and \cref{as:T_assumptions}(2--3)
    hold for batch $t$
    \item $\P(\Pi_{t+1,j} > 0 \mid \Xi_{t} = 1) = 1$ for all $j\in J$.
\end{itemize}
The base case is $t=1$ for $\kappa_2, \Xi_2$. Note that $\hat\Sigma_{tn} \pto \Sigma > 0$.
Under
\eqref{eq:precondition1},
$\hat\Pi_{:1,n} = \frac{n_1}{n}\hat\Pi_{1,n} \pto c_1 \Pi_1 > 0$. Thus, \[
\colvecb{3}{Y_{:1,n}}{\hat\Pi_{:1,n}}{\hat\Sigma_{tn}} \dto \colvecb{3}{\sqrt{c_1}Y_1}
{c_1 \Pi_1}{\Sigma}.
\]
Hence, for $t=1$, \[
\colvecb{2}{\hat\Pi_{:t,n}^{-1} Y_{:t,n} + h_n}{\hat\Sigma_{tn} \hat\Pi_{:t,n}^{-1}} \dto
\colvecb{2}{(c_1)^{-1/2} \Pi_1^{-1} Y_1 + h}{\Sigma \Pi_{1}^{-1}}.
\]
Here, the limit $(c_1)^{-1/2} \Pi_1^{-1} Y_1 + h$ is absolutely continuous with respect to
the Lebesgue measure on indices in $J$, since $Y_1$ is absolutely continuous with respect
to the Lebesgue measure on $\R^K$ and $\Pi_1 > 0$. 

Hence \cref{as:kappa_continuity,as:Xi_continuity} imply that the set of discontinuities
of $\Xi_2(\cdot), \kappa_{2,(1)}(\cdot, \cdot)$ and $\kappa_{2 ,(0)}$ is measure zero with
respect to the law of $((c_1)^{-1/2} \Pi_1^{-1} Y_1 + h, \Sigma \Pi_{1}^{-1})$. Thus, the
continuous mapping theorem applies and $\kappa_2, \Xi_2$ converges. This verifies the
statements in \cref{as:kappa_cmt}(2) and \cref{as:T_assumptions}(2) for $t=1$.

By location-invariance, $\Pi_2, \Xi_2$ are measurable with respect to the differences \[
\br{c_1^{-1/2} \Pi_{1k}^{-1} Y_{1k} - c_1^{-1/2} \Pi_{1K}^{-1} Y_{1K} : k\in [K-1]}.
\]
These differences can be written as $B_1 Y_{1}$ where $B_1 \sqrt{c_1}\Pi_11_K = 0$. This
proves \cref{as:kappa_cmt}(2--3) and \cref{as:T_assumptions}(2--3) for $t=1$.\footnote{It
is easy to show that \[
T_n = \argmax_t \Xi_{t-1,n} - \Xi_{t,n} \dto T = \argmax_t \Xi_{t-1} - \Xi_t
\]
provided that $\Xi_{1:T_0, n} \dto \Xi_{1:T_0}$ by the continuous mapping theorem, since
the support of $(\Xi_{t-1,n} - \Xi_{t,n})_{t=2}^{T_0}$ is the set of standard basis
vectors in
$\R^{T_0}$. 

} Moreover, observe that $\P(\Pi_{2k} >
0) = 1$ whenever $h_{nk} > -\infty$.

In the inductive case, for a generic $t$, assume \eqref{eq:precondition2}. We have that \[
\colvecb{2}{\hat\Pi_{:t,n}^{-1} Y_{:t,n} + h_n}{\hat\Sigma_{tn} \hat\Pi_{:t,n}^{-1}} \dto
\colvecb{2}{\Pi_{:t}^{-1} Y_{:t} + h}{\Sigma \Pi_{:t}^{-1}}.
\]
Note that $\Pi_{:t} > c_1\Pi_{1} > 0$ almost surely. Note that the inductive hypothesis
implies that $\P(\Pi_{tk} > 0 \mid \Xi_{t-1} = 1) = 1$ for $k$ where $h_{nk} > -\infty$.
Note that this means $\br{Y_{:t, j}
: j \in J}$ is
absolutely continuous \emph{conditional on $\Pi_{1:t}, \Xi_t = 1$}
with respect to the Lebesgue measure on $\R^{|J|}$, since the last term $Y_{t} \sim
\Norm(0, \Pi_t\Sigma)$.
As a result, the set of discontinuities of $\Xi_{t+1}, \kappa_{t+1,(0)}, \kappa_{t+1,
(1)}$ is measure zero with respect to the law of $(\Pi_
{:t}^{-1}Y_{:t} + h, \Sigma\Pi_{:t}^{-1}) \mid \Xi_t = 1$. Note that we can write \[
\Xi_{t+1} = \Xi_{t} \Xi_{t+1}.
\] 
Since both $\Xi_{t+1}$ and $\kappa_{t+1}$ are continuous when $\Xi_{t} = 0$, we conclude
that their discontinuities are measure zero with respect to the law of $(\Pi_
{:t}^{-1}Y_{:t} + h, \Sigma\Pi_{:t}^{-1}, \Xi_t)$. Therefore, $\kappa_{t+1}, \Xi_{t+1}$
converges weakly by the continuous mapping theorem.

Similarly, by location-invariance, $\Pi_{t+1}$ is measurable with respect to \[
\br{\Pi_{:t,k}^{-1} Y_{:t,k} - \Pi_{:t,K}^{-1} Y_{:t,K}: k \in [K-1]}.
\]
When written as a transformation $B_t \bY_{1:t}$, we have that \[
B_t \colvecb{3}{\sqrt{c_1}\Pi_1 1_K}{\vdots}{\sqrt{c_{t} } \Pi_t 1_K} = 0.
\]
Lastly, we have that $\P(\Pi_{t+1, j} > 0 \mid \Xi_{t} = 1) = 1$ by assumption. This
completes the proof. 
\end{proof}

\begin{lemma}
\label{lemma:eta}
Under \cref{as:pruning_general} and \cref{as:P_assumptions}, if $\eta_n(\cdot)$ satisfies
\cref{as:eta_pruning}, then it
satisfies
\cref{as:eta_assumptions}.
\end{lemma}
\begin{proof}

Let $\varsigma_{n} = \varsigma(W_{T_n,n})$ denote the strict ranking of the entries of $W_
{T_n, n}$. Likewise, let $E_n =
\br{k :
\Pi_{T_n,
n, k} = 0}$.

For the weak convergence of
$\eta_n$ along a subsequence $n_m$, it suffices to show
that \[
(T_{n_m}, \varsigma_{n_m}, E_{n_m})
\] 
converges weakly along the subsequence, since their support is finite (and hence any
function is continuous under the discrete topology).
Note that since $\varsigma_{n}$ is a function of the indicators $I_{k,k', n}
= \one(
\sqrt{n} W_{T_n,
t, k} - \sqrt{n} W_{T_n,
t, k'} > 0), k\neq k'$ and $k,k' \in [K]$ (and a continuous function since the domain is
discrete). It further suffices to show joint convergence
of \[
(T_{n_m}, \br{I_{k,k', n_m}: k,k' \in [K], k\neq k'}, E_{n_m}). 
\]

Let us consider the above display as some function $\psi_{m}$ of $U_{m} = (T_{n_m},
\hat\Pi_{T_
{n_m}}, \sqrt{n_m}(W_{T_{n_m, n_m}} - \mu_{n_m}))$, where \[
U_m \dto \pr{T, \Pi_{T}, W_T} \quad W_T = \pr{\sum_{t=1}^{T} c_t \Pi_t }^{-1} \sum_{t=1}^T
\sqrt{c_t} Y_t
\]
Note that for any $h \in \R$, $\P(W_{T,k} - W_{T,k'} =
h) = 0$. Along the subsequence, we also have by \cref{lemma:subsequence} that \[
\sqrt{n_m} (\mu_{n_m,k} - \mu_{n_m,k'}) \to \Delta \mu_{k,k'} \in [-\infty, \infty].
\]
Fix some $t_0, \pi_{t_0}, w_{t_0}$ where $w_{t_0,k} - w_{t_0, k'} \neq -\Delta \mu_{kk'}$.
It suffices to show that every sequence $t_m \to t_0, \pi_{t_m} \to \pi_{t_0}, w_{t_m} \to
w_{t_0}$ in the support of $U_m$ has that \[
\psi_n(t_m, \pi_{t_m}, w_{t_m}) \to \pr{
    t_0, \br{\one(w_{t_0, k} - w_{t_0, k'} + \Delta \mu_{kk'} > 0) : k,k' \in [K], k'\neq
    k}, \br{k: \pi_{t_m,k} = 0}
}. 
\numberthis \label{eq:limits_for_eta}
\]
The convergence of $t_m$ is immediate. Moreover, for all sufficiently large $m$, $t_m =
t_0$. Since $w_{t_0,k} - w_{t_0, k'} \neq -\Delta \mu_{kk'}$, the indicator functions
likewise converge. Lastly, note that if $\pi_{t_0,k} = 0$, then for all sufficiently large
$m$, $\pi_{t_m, k} = 0$ as well since $\pi_{t_m, k} \not\in (0, \epsilon)$ by 
\cref{as:pruning_general}; likewise, if $\pi_{t_m, k} > \epsilon$, then for all sufficiently large
$m$, $\pi_{t_m, k} > \epsilon$ as well. Hence $\br{k: \pi_{t_m,k} = 0} = \br{k: \pi_
{t_0,k} = 0}$ for all sufficiently large $m$. This proves \eqref{eq:limits_for_eta}, and
hence it proves the first part of \cref{as:eta_assumptions}.

$\norm{\eta_n} > \epsilon$ is by assumption. Lastly, note that for all sufficiently large
$m$, $\lambda_{n_m} > 0$ and $\hat\Sigma_{n_m} > 0$. Hence, by
\cref{lemma:colspace_rank1}, we have that $\eta_{n_m}$ is in the column space of \[
\frac{\lambda_{n_m} \lambda'_{n_m} / \sqrt{n_m}}{\lambda'_{n_m}1} + \frac{n_T}{n_m}
\hat\Sigma^{-1}_{n_m} \hat \Pi_{Tn_m}
\]
almost surely for all sufficiently large $m$. Similarly for the limit $\eta$.
\end{proof}

\begin{lemma}
\label{lemma:colspace_rank1}
    Let $\R^K \ni u > 0$ be a vector with positive entries. Let $D \ge 0$ be a diagonal
    matrix with
    strictly positive entries on $\emptyset \neq J \subset [K]$. Consider $v\in \R^K$
    whose
    nonzero entries are a subset of $J$. Then $v$ is in the column space of $cuu' +
    D$ for $c > 0$.
\end{lemma}

\begin{proof}
If $J = [K]$, then $cuu' + D$ is full rank. Thus $v$ is in the column space. If $J
\subsetneq [K]$, let $j_0 \in J^C$. Set $q \in \R^K$ where for $j \in J$, $q_j = v_{j}/D_
{jj}$. Let $q_{j_0} = -\sum_{j\in J} u_j q_j$. Let $q_{k} = 0$ for $k \not\in J \cup
\br{j_0}$. By construction, $q'u = 0$ and $Dq = v$. This concludes the proof.
\end{proof}

\section{Two-batch experiment and properties of \citet{zhang2020inference}}
\label{asub:zhang_bet_proof}

To illustrate properties of \citet{zhang2020inference}'s procedure, we restrict to a
simple two-batch setup. Consider \[
X_1 \sim \Norm(\mu, 2I) \text{ and } X_2 \mid X_1 \sim \Norm(\mu, \diag\pr{1/\Pi,
1/(1-\Pi)}).
\]

Since $\Pi_2 = (\Pi_{21}, 1-\Pi_{21})$ is summarized by $\Pi_{21}$, we drop the subscript
and refer to $\Pi = \Pi_{21}$. Note that if the assignment algorithm is
location-invariant---that is, $\Pi(X_1) = \Pi_t (X_1 + h1_2)$ for every $h$---then
$\Pi_t(X_1) = \Pi_t(X_{11}, X_{12}) = \Pi_t(X_{11} - X_{12},0) = \Pi_t(X_{11} - X_{12})$
depends only on the difference $\Delta \equiv X_{11} - X_{12}$.

Suppose we are interested in making inferences about $\mu_1$. \citet{zhang2020inference}
propose a simple inference procedure. The Gaussian limit analogue of their procedure is
based on the following statistic, which is a sum of studentized batch means: \[Z =
\frac{1}{\sqrt{2}} X_{11} + \Pi(\Delta)^{1/2} X_{21}.
\]
Since the studentized second batch statistic is independent
of $\Delta$,\[
\Pi^{1/2} (X_{21} - \mu_1) \mid \Delta \sim \Norm(0, 1),
\]
the following quantity is pivotal: \[
Z^* = Z^*(\mu_1) = Z - \pr{\frac{1}{\sqrt{2}} + \Pi^{1/2}} \mu_1 \sim \Norm(0, 2).
\]
Thus, for example, we can test $H_0 : \mu_1 = \mu_{01}$ by comparing $Z^* (\mu_{01})$ to
its distribution $\Norm(0,2)$ under $H_0$. As another example, the natural estimator
$\pr{\frac{1}{\sqrt{2}} + \Pi^{1/2}}^{-1} Z$ is median-unbiased for $\mu_1.$

Note that the joint density of $X_1, X_2$ factors into
\begin{align*} p_\mu(X_1, X_2) &= f_1(U \mu_1)f_2(V \mu_2)f_3(\Pi(\Delta) (\mu_1^2 -
\mu_2^2)) h_1(\mu) h_2(X_1, X_2)
\end{align*}
where $U = \frac{1}{2} X_{11} + \Pi X_{21}$, $V = \frac{1}{2} X_{12} +
(1-\Pi) X_{22}$, $h_1$ is free of $(X_1,X_2)$, and $h_2$ is free of $\mu$. As a
result, the statistics $(U, V, \Pi)$ are sufficient for $\mu$ in the law of $(X_1, X_2)$.
Consequently, $(U, V, \Delta)$ is sufficient as well. We can decompose the statistic $Z$
as a linear combination of the sufficient statistics $(\Delta, U, V)$ and a Gaussian noise
term. We do so by studying the joint Gaussian distribution of $(U, V, Z)$ conditional on
$\Delta$.

\begin{restatable}{prop}{proptwostage}
\label{prop:twostage}
    The statistic $Z$ obeys the following representation: \begin{align*}
    Z &= a\Delta + b(\Pi) V + c(\Pi) U + \sigma(\Pi) \xi
    \\
    &=\frac{\Delta}{2\sqrt{2}} + \frac{\sqrt{2} \Pi - \sqrt{\Pi}}{1 + 4\Pi - 4\Pi^2} V + \frac{\sqrt{2} + 5\sqrt{\Pi} - \sqrt{2}\Pi - 4 \Pi^{3/2}}{1+4\Pi - 4\Pi^2} U  \\ &\quad+ \sqrt{
\frac{(1+2\Pi - 2\sqrt{2\Pi})(1-\Pi)}{1 + 4\Pi - \Pi^2}
} \cdot \xi
    \end{align*}
    for some $\xi \sim \Norm(0,1)$ independently of $(U, V, \Delta)$.
\end{restatable}

\begin{proof}
    Observe that \[
    \colvecb{3}{\bar X_1}{X_{21}}{X_{22}} \mid \underbrace{X_{11} - X_{12}}_{\Delta} \sim
    \Norm \pr{
        \colvecb{3}{\frac{\mu_1 + \mu_2}{2}}{\mu_1}{\mu_2}, \begin{bmatrix}
            1 & 0 & 0 \\
            0 & 1/\Pi & 0 \\
            0 & 0 & 1/(1-\Pi)
        \end{bmatrix}
    }
    \]
    where $\bar X_1 = \frac{1}{2} (X_{11} + X_{12})$. Note that, since $X_{11} = \bar X_1 +
    \Delta/2$ and $X_{12} = \bar X_1 - \Delta/2$, the above display captures the distribution
    of
    the data conditional on $\Delta$. Hence we can compute that \[
    \colvecb{3}{Z-\frac{\Delta}{2\sqrt{2}}}{U - \frac{\Delta}{4}}{V + \frac{\Delta}{4}} \mid
    \Delta \sim \Norm\pr{
        \colvecb{3}{\frac{\mu_1+\mu_2}{2\sqrt{2}} +
        \sqrt{\Pi}\mu_1}{\frac{\mu_1+\mu_2}{4} + \Pi \mu_1}{
        \frac{\mu_1+\mu_2}{4} + (1-\Pi) \mu_2}
    , \begin{bmatrix}
        \frac{3}{2} & \frac{1}{2\sqrt2} + \sqrt{\Pi} & \frac{1}{2\sqrt{2}} \\
        \frac{1}{2\sqrt2} + \sqrt{\Pi} & \frac{1}{4} + \Pi & \frac{1}{4} \\
        \frac{1}{2\sqrt{2}} & \frac14 & \frac54 - \Pi
    \end{bmatrix}},
    \]
    since the left-hand side are functions of $\bar X_1, X_{21}, X_{22}$ and $\Delta$. The
    claim follows by computing that \[
    \E[Z \mid \Delta, U, V] = a \Delta + b(\Pi) V + c(\Pi) U
    \]
    and $\var(Z \mid \Delta, U, V) = \sigma^2(\Pi)$, since $Z,U,V$ are jointly Gaussian
    conditional on $\Delta$.
    \end{proof}

One can further show that the projection of the pivot $Z^*$ onto $(\Delta, V, U)$ is
Gaussian with conditional mean $a(\Delta + \mu_2 - \mu_1)$ and conditional variance $
\frac{3}{2} - \sigma^2(\Pi) $:
\[a \Delta + b(\Pi)V + c(\Pi) U - \pr{\frac{1}{\sqrt{2}} +
\Pi^{1/2}} \mu \mid \Delta \sim \Norm\pr{ a(\Delta + \mu_2 - \mu_1), \frac{3}{2} -
\sigma^2(\Pi) }\] As a result, the unconditional distribution of the projection of the
pivot onto $(\Delta, U, V)$ is a mean and variance mixture of Gaussians, whose
unconditional mean is zero. The pivot $Z^*$ is thus Gaussian by injecting
$\Delta$-dependent noise $\sigma(\Pi)\xi$ so as to stabilize the $\Delta$-conditional
variance.

Due to the presence of the noise term $\xi$, the statistic $Z$ suffers from certain
deficiencies. For instance, there exist \emph{recognizable subsets}---events on which
inferences based on $Z$ are unreliable.

\begin{restatable}{cor}{cortwostageone}
\label{cor:twostage1}
Assume that $\Pi \neq \frac{1}{2}$ with positive probability and $\Pi \in (0,1)$ almost
surely.
    Fix any significance level $\alpha \in (0,1/2)$. There exists an event $A$ which does
    not depend on the unknown parameter $\mu$ on which the coverage probability is at most
    $1/2$. That is, for all $\mu \in \R^2$, $\P_\mu(A) > 0$, and  \[
\P_\mu\pr{\abs[\bigg]{Z - \pr{\frac{1}{\sqrt{2}} + \Pi^{1/2}} \mu} \le \sqrt{2} \Phi^{-1}
(1-\alpha/2) \mid A} \le \frac{1}{2}.
    \]
\end{restatable}

\begin{proof}
    Let $A = \br{|\xi| > \sqrt{2} \frac{1}{\sigma(\Pi)} \Phi^{-1} (1-\alpha/2), \Pi \neq
    1/2}$. Since when $\Pi \neq 1/2$, $\sigma(\Pi) > 0$, and $\Pi \neq 1/2$ with positive
    probability, $A$ occurs with positive probability. Consider $Z = a\Delta + b(\Pi)V
    + c(\Pi) U + \sigma(\Pi) \xi$ and $Z = a\Delta + b(\Pi)V
    + c(\Pi) U - \sigma(\Pi) \xi$. Note that on the event $A$, \[
    |Z-Z'| > 2\sqrt{2} \Phi^{-1} (1-\alpha/2).
    \]
    Hence, on $A$, \[
    \abs{Z - (1/\sqrt{2} + \Pi^{1/2}) \mu}  < \sqrt{2} \Phi^{-1}(1-\alpha/2) \text{ and } \abs{Z' - (1/\sqrt{2} + \Pi^{1/2}) \mu}  < \sqrt{2} \Phi^{-1}(1-\alpha/2)
    \]
    are mutually exclusive events. Moreover, the probabilities of these two events conditional
    on $A$ are identical by symmetry. Therefore, \[
    2\P\bk{\abs{Z - (1/\sqrt{2} + \Pi^{1/2}) \mu}  < \sqrt{2} \Phi^{-1}(1-\alpha/2) \mid A}
    \le 1.
    \]
    This completes the proof.
    \end{proof}

Moreover, the natural estimator $\pr{\frac{1}{\sqrt{2}} + \Pi^{1/2}}^{-1} Z$ can be
improved in mean-squared error by Rao--Blackwellization against $(U, V, \Delta)$:
\begin{restatable}{cor}{cortwostagetwo}
\label{cor:twostage2}
Let $T_0 = \pr{\frac{1}{\sqrt{2}} + \Pi^{1/2}}^{-1} Z$.
    Let \[T^* = \E\bk{T_0 \mid U, V,
    \Delta} = \pr{\frac{1}{\sqrt{2}} + \Pi^{1/2}}^{-1}\pr{a\Delta + b(\Pi) V + c(\Pi)
    U}\]
    be its conditional expectation with respect to $U, V, \Delta$.
    Then, for any $\mu \in \R^2$, \[
\E_\mu[(T^* - \mu_1)^2] \le \E_{\mu}[(T_0 - \mu_1)^2].
    \]
\end{restatable}

\begin{proof}
    Note that we can write $T_0 = T^*(U, V, \Delta) + q(\Pi) \xi$ for some $q(\Pi)$. Note that
    \[
    \E[(T_0 - \mu_1)^2] = \E[q(\Pi)^2 \xi^2] + \E[(T^* - \mu_1)^2] \ge \E[(T^* - \mu_1)^2].
    \]
\end{proof}

\end{document}